\documentclass[12pt,draftcls,onecolumn]{IEEEtran}
\usepackage{amsmath,amssymb,graphicx,epstopdf,enumitem,algorithm,algorithmic,pgfplots,subfig,cite,amsthm}
\usepackage{balance}
\usepackage{acronym}
\usetikzlibrary{shapes,arrows,positioning,patterns}

\long\def\comment#1{}

\newfont{\bbb}{msbm10 scaled 700}

\newfont{\bb}{msbm10 scaled 1100}

\newcommand{\rv}{{\bf r}}

\newcommand{\Am}{{\bf A}}

\newcommand{\Xm}{{\bf X}}

\newcommand{\Lc}{{\cal L}}
\newcommand{\Mc}{{\cal M}}
\newcommand{\Nc}{{\cal N}}

\newcommand{\Xc}{{\cal X}}

\renewcommand{\arg}{{\hbox{arg}}}

\newenvironment{varalgorithm}[1]
  {\algorithm[t!]\renewcommand{\thealgorithm}{#1}}
  {\endalgorithm}

\newlength\figureheight
\newlength\figurewidth

\newcommand\notype[1]{\unskip}

\newtheorem{prop}{Proposition}
\newtheorem{remk}{Remark}

\acrodef{ATG}{air-to-ground}
\acrodef{SC}{small cell}
\acrodef{NFP}{networked flying platform}
\acrodef{GAP}{generalized assignment problem}
\acrodef{mmWave}{millimeter-Wave}
\acrodef{FSO}{free space optics}
\acrodef{C-RAN}{cloud radio access network}
\acrodef{CAPEX}{capital expenditure}
\acrodef{NLoS}{non-line-of-sight}
\acrodef{LoS}{line-of-sight}
\acrodef{HetNet}{heterogeneous network}
\acrodef{LAP}{low-altitude platform}
\acrodef{MAP}{medium-altitude platform}
\acrodef{HAP}{high-altitude platform}
\acrodef{FSPL}{free-space path loss}
\acrodef{QoS}{quality-of-service}
\acrodef{LP}{Linear Programming}

\begin{document}

\title{Small Cell Association with Networked Flying Platforms: Novel Algorithms and Performance Bounds}

\author{Syed~Awais~Wahab~Shah,
        Tamer~Khattab,~\IEEEmembership{Senior Member,~IEEE,}
        Muhammad~Zeeshan~Shakir,~\IEEEmembership{Senior Member,~IEEE,}
        Mohammad~Galal~Khafagy,~\IEEEmembership{Senior Member,~IEEE,}
        and~Mazen~Omar~Hasna,~\IEEEmembership{Senior Member,~IEEE}%
\thanks{S. A. W. Shah, T. Khattab, M. G. Khafagy and M. O. Hasna are with the Electrical Engineering department, Qatar University, Doha, Qatar (e-mail: \{syed.shah, tkhattab, mkhafagy, hasna\}@qu.edu.qa). M. Z. Shakir is with the School of Engineering and Computing, University of the West of Scotland, Paisley, Scotland, UK (e-mail: muhammad.shakir@uws.ac.uk).}}


\maketitle

\begin{abstract}
Fifth generation (5G) and beyond-5G (B5G) systems expect coverage and capacity enhancements along with the consideration of limited power, cost and spectrum. Densification of \acp{SC} is a promising approach to cater these demands of 5G and B5G systems. However, such an ultra dense network of \acp{SC} requires provision of smart backhaul and fronthaul networks. In this paper, we employ a scalable idea of using \acp{NFP} as aerial hubs to provide fronthaul connectivity to the \acp{SC}. We consider the association problem of \acp{SC} and \acp{NFP} in a \ac{SC} network and study the effect of practical constraints related to the system and \acp{NFP}. Mainly, we show that the association problem is related to the \ac{GAP}. Using this relation with the \ac{GAP}, we show the NP-hard complexity of the association problem and further derive an upper bound for the maximum achievable sum data rate. Linear Programming relaxation of the problem is also studied to compare the results with the derived bounds. Finally, two efficient (less complex) greedy solutions of the association problem are presented, where one of them is a distributed solution and the other one is its centralized version. Numerical results show a favorable performance of the presented algorithms with respect to the exhaustive search and derived bounds. The computational complexity comparison of the algorithms with the exhaustive search is also presented to show that the presented algorithms can be practically implemented.
\end{abstract}

\begin{IEEEkeywords}
5G, backhaul/fronthaul network, binary integer linear program, drones, networked flying platforms (NFPs), small-cell networks, Unmanned aerial vehicles (UAVs)
\end{IEEEkeywords}

\IEEEpeerreviewmaketitle
\acresetall

\section{Introduction}\label{sec:intro}
\IEEEPARstart{T}{he} advancement of technology (such as video services) and a rapid growth in the number of cellular users (such as mobile devices and tablets etc.,) have been pushing the limits of wireless communication systems. Next generation systems expect coverage and capacity enhancements along with the consideration of limited power, cost and spectrum. To cater for these demands, a ten-fold increase in the radio spectrum is required \cite{pizzinat2015}. Therefore, researchers both in academia and industry are looking towards latest wireless technologies such as \ac{mmWave} and \ac{FSO}, as they can provide hundreds of megahertz of bandwidth for wireless transmission. However, these wireless technologies under the usual power constraints have a limited range as the signal degrades due to environmental effects. This transmitter-receiver distance reduction and the growing cellular user crowds lead to the idea of \ac{SC} densification. This densification of \acp{SC} (e.g., pico and femto cells) is being considered as a corner stone of fifth generation (5G) and beyond-5G (B5G) cellular networks.

China Mobile in 2011 \cite{chih2014toward} proposed a \ac{C-RAN} architecture, which is considered as a promising paradigm for 5G and B5G cellular systems, as it can resolve the backhaul traffic limitations by providing a fronthaul link. Due to the dense deployment of \acp{SC}, fronthaul links demand a high capacity of more than 2.5 Gbps with a low latency of around 100 $\mu$s or less \cite{Ericsson_m2020}. In terms of wired technology, these demands can be fulfilled only by fiber optical links as they offer an abundant bandwidth with low latency data transfer. However, such fiber links deployment results in high \ac{CAPEX} as compared to wireless fronthaul links \cite{SC_virt15}. Wireless fronthaul links can be realized using microwave bands for \ac{NLoS} case or \ac{mmWave}/\ac{FSO} for \ac{LoS} case. Microwave links can cover a wide area but suffer from low data rates as currently available commercial products provides a maximum of 2 Gbps throughput \cite{SC_virt15}. \ac{FSO} and \ac{mmWave} based fronthaul links have attracted an eye of various researchers as they meet the capacity requirements of 5G and B5G systems and they are light-weight and easy to install. However, \ac{mmWave}/\ac{FSO} suffer from susceptibility to weather conditions \cite{ShakirFSOMAG} and require a \ac{LoS} connection, which is a main hurdle in urban regions due to few available ground locations. Recently, a scalable idea was presented in \cite{ShakirFSOMAG} that utilizes \acp{NFP} as a wireless fronthaul hub point between \acp{SC} and core network. These \ac{NFP}-hubs provide a possibility of wireless \ac{LoS} fronthaul link to utilize radio frequency (RF), \ac{mmWave} and \ac{FSO} technologies, and thus, overcomes the limitations of few available wireless \ac{NLoS} ground fronthaul links.

Recently, both academia and industry started taking interest in utilizing \acp{NFP} such as unmanned aerial vehicles (UAVs), drones and unmanned balloons for wireless communications. These \acp{NFP} can be manually controlled but mainly designed for autonomous pre-determined flights. Latest \acp{NFP} are capable of carrying RF/\ac{mmWave}/\ac{FSO} payloads along with an extended battery life \cite{ahmadi2017novel}. On the basis of their flying altitude range, \acp{NFP} are categorized into \ac{LAP} (less than 5km), \ac{MAP} (between 5km to 10km) and \ac{HAP} (greater than 10km).

In this work, we employ \acp{NFP} as aerial hubs to provide fronthaul connectivity to a network of \acp{SC}. We define the association problem of \acp{SC} and \acp{NFP}, present its performance bounds, then propose novel efficient (less computationally complex) centralized and distributed greedy algorithms for its solution.

\subsection{Related Work}\label{sec:Rel_Work}
With the popularity of \acp{NFP}, a widely used \ac{ATG} propagation model was presented in \cite{ATGmodel}. This model considers the aerial communication between \acp{NFP} and terrestrial nodes. Later on, a closed form expression of the path loss and the effect of change of altitude of the \ac{NFP} over the coverage area was presented in \cite{ATG_optDrone1}. For the case of two \acp{NFP}, the coverage area was analyzed by varying the distance between the \acp{NFP} and their altitudes in \cite{TwoDrones}.

In the literature, 3D placement of \acp{NFP} and a related research problem of the association of \acp{NFP} and users were studied by a few researchers \cite{IremOneDrone, ElhamBackhaul, alzenad20173d, ElhamMultiPSO, Mozaffari2016, Sharma2016, mozaffari2017OTT, Mozaffari2017OTTTrans, kalantari2017user}. In all of those works, the \acp{NFP} were used as flying base stations (BSs) to provide wireless connectivity to the ground users. In \cite{IremOneDrone}, authors have designed the 3D placement problem of a single \ac{NFP} BS considering only the signal-to-noise ratio (SNR) as a \ac{QoS} parameter and studied the coverage region of the \ac{NFP} BSs for different urban environments. A number of constraints including backhaul data rate, maximum bandwidth of a single \ac{NFP} and path loss were taken into consideration for joint 3D placement and association problem of a single \ac{NFP} BS in \cite{ElhamBackhaul}. However, a computationally-expensive and not practically-implementable exhaustive search method was used in both \cite{IremOneDrone} and \cite{ElhamBackhaul} to solve the designed problems. In \cite{alzenad20173d}, the 3D placement problem was decoupled into first finding the optimal altitude and then using circle placement problem to optimize the 2D placement of a single \ac{NFP} BS in order to maximize the number of users in a covered region.

For the case of multiple NFPs, association of NFP BSs and users on the basis of SINR parameter was presented in \cite{ElhamMultiPSO}, then the 3D placement problem is solved using particle swarm optimization (PSO) algorithm. The work in \cite{Mozaffari2016} and \cite{Sharma2016} dealt with the 3D placement of the \acp{NFP} considering only the SINR constraint, where the former used circle packing theory to enhance the coverage performance with minimum power, while the later used entropy and network bargaining approaches to enhance the capacity and coverage area. In \cite{mozaffari2017OTT} and \cite{Mozaffari2017OTTTrans}, a delay-sensitive cell association problem was designed for multiple \ac{NFP} BSs that co-exist with terrestrial BSs and optimal packing theory was used to solve the designed problem. A \ac{LP} relaxation along with rounding was used in \cite{kalantari2017user} to solve the association problem and then the PSO algorithm was utilized to solve the 3D placement problem. However, due to \ac{LP} relaxation and then rounding, a number of constraints of the association problem may not be satisfied exactly.

Since \acp{NFP} gain popularity in communication systems, they have been studied as either repeaters or BSs to enhance the network coverage and signal strength mainly in hard to reach areas. To the best of our knowledge, there is only one work, \cite{ShakirFSOMAG}, in the literature apart from our recently published conference papers, \cite{AwaisGlobeCom} and \cite{AwaisPIMRC}, that uses \acp{NFP} as hub points to provide fronthaul connectivity. The work in \cite{ShakirFSOMAG} was limited to the feasibility study of using \acp{NFP} as fronthaul hubs, design of backhaul framework, investigation about the effect of weather conditions on the system and evaluation of the implementation cost of the proposed system as compared to other wired/wireless fronthaul/backhaul links. In \cite{AwaisGlobeCom} and \cite{AwaisPIMRC}, we formulated and analyzed the association problem of \acp{SC} and \acp{NFP}. Further, instead of using exhaustive search methods, we presented efficient greedy algorithms.

\begin{figure*}[t!]\centering
    \includegraphics[width=10cm]{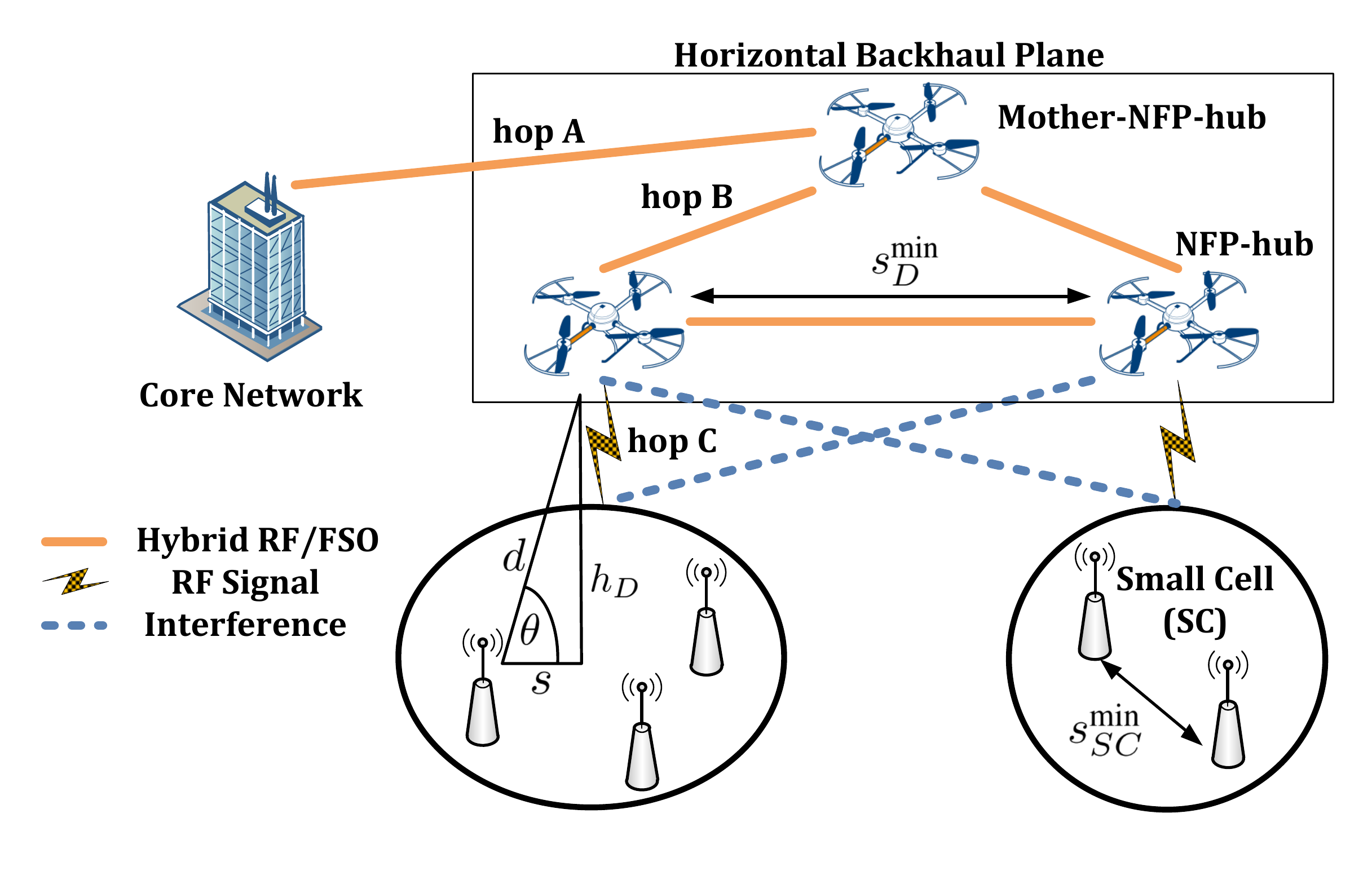}
	\caption{Graphical illustration of HetNet of \acp{SC} and \acp{NFP} communication system.}
	\label{fig:SysMod}
\end{figure*}

\subsection{Contributions}\label{sec:Contribution}
This work is an extension of our work in \cite{AwaisGlobeCom} and \cite{AwaisPIMRC}. Here, we reconsider the mathematical problem formulation for the association of \acp{SC} and \acp{NFP}. We also study in detail the effect of a number of practical constraints on the association problem, where the constraints are related to 5G and B5G systems and \acp{NFP}, such as backhaul data rate, \ac{NFP}'s bandwidth and number of links limitations. We further modify our previously proposed greedy algorithms to achieve enhanced problem solutions.

In this work, we propose an analytical framework for the analysis of the association problem of \acp{SC} and \acp{NFP}. On top of the proposed framework, the main contributions of the paper are summarized as follows:
\begin{enumerate}[label=\roman*)]
  \item We show the relevance of the association problem of \acp{SC} and \acp{NFP} with the \ac{GAP}. Using this relevance, for the first time in the literature to the best of the authors' knowledge, we show that the association problem is at least NP-hard.
  \item Again, capitalizing on the relation with the \ac{GAP} problem, we present an analytical derivation of the upper bound for the association problem with some relaxations. This follows the same framework used for a well known branch and bound (B\&B) method for \ac{GAP} \cite{rossBound}.
  \item We present efficient (less complex) greedy solutions for the association problem as opposed to the exhaustive search presented in literature for related problems\footnote{The algorithms presented here are further modified and enhanced versions of our algorithms presented in \cite{AwaisGlobeCom} and \cite{AwaisPIMRC}.}.
\end{enumerate}

\subsection{Paper Organization and Notations}\label{sec:org_not}
The rest of the paper is organized as follows. In section \ref{sec:Sys_Prob}, a system model and the association problem are presented. Section \ref{sec:Prob_Approx} includes the relation of association problem with the GAP, its approximation and the upper bounds using the same relevance. Section \ref{sec:Prop_Sol} presents two efficient greedy solutions proposed for the association problem. Numerical results and related discussions are presented in section \ref{sec:Sim_Res}. Computational complexity of the algorithms is discussed in section \ref{sec:Comp_Complexity} and finally section \ref{sec:Conc} concludes the paper.

Following notations are used in the paper. A constant number is denoted as either $x$ or $X$. The matrix and its ($i,j$)-th entry are denoted by $\Xm$ and $x_{ij}$ or $X_{ij}$, respectively. A set or a list is represented as $\Xc$. Notation $\boldsymbol{0}_{A \times B}$ represents an $A \times B$ matrix of all 0 entries.

\section{System Model and Problem Formulation}\label{sec:Sys_Prob}
This section first presents the \ac{NFP} connected \acp{SC} system model under consideration. An \ac{ATG} channel model is provided that highlights mainly the path loss parameter for the communication between \acp{NFP} and \acp{SC}. Finally, the association problem of \acp{SC} and \acp{NFP} is emphasized and mathematically formulated for the user-centric case, where the objective is to maximize the overall sum data rate.

\subsection{System Model}\label{sec:Sys_Model}
Consider a \ac{HetNet} (e.g., a 5G or B5G system) as shown in Fig. \ref{fig:SysMod} that consists of three classes of wireless nodes: i) ground \acp{SC}, ii) \ac{NFP}-hubs, and iii) ground core network gateway. \acp{SC} accumulate and route the downlink/uplink traffic between cellular users and core network using fronthaul links. \acp{NFP} act as hub points to provide fronthaul connectivity between \acp{SC} and core network. For brevity, \ac{NFP}-hubs will be referred to as \acp{NFP}. \acp{NFP} are distributed in a two-level hierarchy, where a number of \acp{NFP} spread over a region up to an altitude of 5km, i.e., \ac{LAP}, and \acp{NFP} are connected to a mother \ac{NFP} placed at an altitude of higher than 5km, i.e., either \ac{MAP} or \ac{HAP} \cite{ahmadi2017novel}.

\acp{NFP} are connected to each other and mother-\ac{NFP} through \ac{FSO} links, where we neglect the \ac{FSO} link losses in this work. \acp{NFP} can share the control information such as bandwidth, data rate and other requirements with each other as well as mother-\ac{NFP}, however, all the data information can only be shared with mother-\ac{NFP}. We assume that the distribution of \acp{SC} and \acp{NFP} does not change for time duration $T$, and thus, we study their association considering the active \acp{SC} and \acp{NFP} during the time interval $\begin{bmatrix} 0 & T \end{bmatrix}$.

\subsection{Air-to-Ground Path Loss Model}\label{sec:ATG_Model}
For the communication between \acp{NFP} and \acp{SC}, we have adopted a widely used \ac{ATG} path loss model presented in \cite{ATGmodel} and \cite{ATG_optDrone1}. The model is based on the proposition supported by the statistical derivation in \cite{ATGmodel} that the \ac{ATG} communications may belong to one of the two propagation groups: i) \ac{LoS} receivers, and ii) \ac{NLoS} receivers. The first propagation group includes the receivers placed in \ac{LoS} or near-\ac{LoS} conditions, however the \ac{NLoS} receivers rely on the coverage via reflections and diffractions only. The radio signals first propagate through the free space incurring \ac{FSPL} and then reaches the receivers either directly (i.e., \ac{LoS} receivers) or incur scattering and shadowing because of man-made structures (i.e., \ac{NLoS} receivers). The two propagation groups result in a path loss (referred as excessive path loss that is additional to \ac{FSPL}) following a Gaussian distribution \cite{ATG_optDrone1}. The considered model in \cite{ATG_optDrone1} deals with its mean value instead of its random behavior.

An important factor of the mean value of the excessive path loss is the probability of \ac{LoS} $\left(P(\text{LoS})\right)$, that depends on the considered environment (such as rural, urban, or others) and the orientation of \acp{NFP} and ground \acp{SC} and it was formulated in \cite{ATGmodel} and \cite{ATG_optDrone1} as
\begin{equation}\label{Prob_LoS}
    P(\text{LoS}) = \frac{1}{1 + \alpha \exp \left\{ -\beta \left( \frac{180}{\pi} \theta - \alpha \right) \right\}}
\end{equation}
where $\alpha$ and $\beta$ are parameters with constant values that depend on the specific environment. The elevation angle from the ground \ac{SC} to the \ac{NFP} is represented by $\theta = \arctan \left( \frac{h_D}{s} \right)$, where $s = \sqrt{ \left( x - x_D \right)^2 + \left( y - y_D \right)^2 }$ denotes the horizontal distance between \ac{SC} and \ac{NFP}. The positions of \acp{SC} and \acp{NFP} in a cartesian coordinate system with respect to the origin are denoted by $\left(x,y\right)$ and $\left(x_D,y_D,h_D\right)$, respectively. The mean path loss is presented as
\begin{equation}\label{PathLoss}
    \text{PL}(dB) = 10 \; \log \left( \frac{4 \pi f_c d}{c} \right)^\gamma + P(\text{LoS}) \eta_{\text{LoS}} + P(\text{NLoS}) \eta_{\text{NLoS}}
\end{equation}
where the first term represents the \ac{FSPL} that depends on carrier frequency $f_c$, speed of light $c$, path loss exponent $\gamma$ and the distance $d=\sqrt{h_D^2+s^2}$ between \ac{NFP} and \ac{SC}. Variables $\eta_{\text{LoS}}$ and $\eta_{\text{NLoS}}$ represent additional losses for \ac{LoS} and \ac{NLoS} links, respectively and $P(\text{NLoS}) = 1 - P(\text{LoS})$. All parameters in \eqref{PathLoss} depend on the environment. It can be noticed from \eqref{PathLoss} that for a known distribution of \acp{SC} and \acp{NFP}, if we fix the PL then we can estimate the geographical area covered by the \ac{NFP} \cite{ATG_optDrone1}.

\subsection{Problem Formulation}\label{sec:Prob_Form}
The communication of the user data between \acp{SC} and the core network depends on the fronthaul link of the \acp{NFP}. An intelligent association and placement of the \acp{NFP} can provide efficient throughput, widespread connectivity and result in a better \ac{QoS}. In this work, we fix the height of the \acp{NFP} and consider a random 2D placement of \acp{NFP} and \acp{SC}. Thus, our main focus is the association aspect of the problem.

Consider the system shown in Fig. \ref{fig:SysMod}, where $N_{SC}$ \acp{SC} are distributed randomly over a square region of area $A_S$. Over the same region, $N_D$ \acp{NFP} are distributed randomly in a horizontal plane at an altitude of $h_D$ above the ground level. Mother-\ac{NFP} is placed at a height greater than $h_D$, so it can have a direct \ac{LoS} connection with the $N_D$ \acp{NFP}. Using a stochastic-geometry approach, the random distribution of both \acp{SC} and \acp{NFP} follows a \emph{Mat\`ern} type-I hard-core process \cite{matern2013} with the same average density of $\lambda$ per $\text{m}^2$ having a minimum separation of $s_{SC}^{\text{min}}$ and $s_D^{\text{min}}$ with their neighbors, respectively. Note that, the average number of \acp{SC} and \acp{NFP} in a given area is equal to their average density multiplied by the size of the area such that $N_{k}^{avg} = \lambda \exp(-\lambda A_{k_{\text{sep}}}) A_S$, where the subscript $k \in \{SC, D\}$ refers to \acp{SC} and \acp{NFP}, respectively. $A_{k_{\text{sep}}}$ is the area of the separation, i.e., $A_{k_{\text{sep}}} = {s_{k}^{\text{min}}}^2$ for a square region. Let us denote the random distribution points of both \acp{SC} and \acp{NFP} as $\left(x_i,\,y_i\right)$ and $(x_{D_j}, y_{D_j}, h_{D_j})$, respectively, where $i \in \Nc = \left\{ 1, \ldots, N_{SC} \right\}$ and $j \in \Mc = \left\{ 1, \ldots, N_D \right\}$.

First of all, before studying the objective of the association problem, the limiting factors (related to the available resources) that affect the communication between \acp{SC} and \acp{NFP} are discussed below. Three major limiting factors have a direct effect on the association of \acp{SC} and \acp{NFP}, namely, backhaul data rate, bandwidth, and the number of links of the \acp{NFP}. Thus, the association between \acp{SC} and \acp{NFP} varies with changes in the limiting factors.

The backhaul link between the core network and the mother-\ac{NFP}, i.e., hop A, limits the maximum allowable data rate of the network, that is referred here as backhaul data rate $R$. This means that the sum of the data rate for all the \ac{NFP} and \ac{SC} pairs cannot exceed the backhaul data rate $R$. Let us denote the requested data rate of $i$-th \ac{SC} associated with $j$-th \ac{NFP} by $r_{ij}$, then this constraint can be written as \eqref{cons1}, where $A_{ij}$ is an entry of an $N_{SC} \times N_D$ association matrix $\Am$ that shows the association of \acp{SC} and \acp{NFP} as
\begin{equation}\label{Assoc_Mat}
      A_{ij} =
        \left\{
          \begin{array}{ll}
            1, & \hbox{if $i$-th SC is connected with $j$-th \ac{NFP},} \\
            0, & \hbox{otherwise.}
          \end{array}
        \right.
\end{equation}

The next limitation is posed by the fronthaul \ac{FSO} link in the hop B, i.e., from mother-\ac{NFP} to each \ac{NFP}. Depending upon the quality of the \ac{FSO} link, the $j$-th \ac{NFP} is allocated a maximum bandwidth, $B_j$, that can be distributed among associated \acp{SC}. This limits the sum of requested bandwidth of all \acp{SC} associated with $j$-th \ac{NFP} and it can be mathematically represented as \eqref{cons2}. The allocated bandwidth $b_{ij} = \frac{r_{ij}}{\eta_{ij}}$ of the $i$-th \ac{SC} and the $j$-th \ac{NFP} pair depends on $r_{ij}$ and the spectral efficiency $\eta_{ij} =$ $\log_2 \left( 1+\text{SINR}_{ij} \right)$, where the SINR can be expressed as
\begin{equation}\label{eq:SINR}
      \text{SINR}_{ik} = \frac{P_{r_{ik}}}{\sum_{j=1,j \neq k}^{N_D} P_{r_{ij}} + \sigma}
\end{equation}
Here, $P_{r_{ij}}$ represents the received power from the $j$-th \ac{NFP} to the $i$-th \ac{SC} and $\sigma$ represents the noise floor of each receiver.

In the next hop, i.e., hop C between the $j$-th \ac{NFP} and the $i$-th \ac{SC}, the RF fronthaul link should satisfy a \ac{QoS} requirement. Every \ac{NFP} can serve \acp{SC} placed inside a specific area computed using \eqref{PathLoss} for fixed positions of \acp{NFP}, \acp{SC} and a maximum path loss \cite{ATG_optDrone1} and \cite{TwoDrones}. This maximum path loss is dictated by the minimum required SINR to serve a \ac{SC} via RF link. Thus, each \ac{NFP}-\ac{SC} pair link should satisfy a minimum SINR \ac{QoS} requirement that can be written as \eqref{cons3}.

Considering all the above mentioned constraints, for fixed positions of \acp{NFP} and \acp{SC}, Our objective is to find the best possible association of the \acp{SC} with the \acp{NFP} such that the sum data rate of the overall system is maximized. Such a problem can be formulated as
\begin{subequations}
\label{eq:Opt_Prob}
\begin{alignat}{3}
    \max_{\{A_{ij}\}} \quad \sum_{i=1}^{N_{SC}} &\sum_{j=1}^{N_D} r_{ij} \cdot A_{ij} \label{Obj_Fun}\\
    \intertext{subject to}
    \sum_{i=1}^{N_{SC}} \sum_{j=1}^{N_D} r_{ij} \cdot A_{ij} \; &\leq \; R, \label{cons1}\\
    \sum_{i=1}^{N_{SC}} b_{ij} \cdot A_{ij} \; &\leq \; B_j, && \quad \text{for } j \in \Mc, \label{cons2}\\
    \frac{1}{\text{SINR}_{ij}} \cdot A_{ij} \; &\leq \; \frac{1}{\text{SINR}_{\text{min}}}, && \quad \text{for } i \in \Nc \; \& \; j \in \Mc, \label{cons3}\\
    \sum_{i=1}^{N_{SC}} A_{ij} \; &\leq \; N_{l_j}, && \quad \text{for } j \in \Mc, \label{cons4}\\
    \sum_{j=1}^{N_D} A_{ij} \; &\leq \; 1, && \quad \text{for } i \in \Nc, \label{cons5}\\
    A_{ij} \; &\in \; \left\{ 0, 1 \right\}, && \quad \text{for } i \in \Nc \; \& \; j \in \Mc. \label{cons6}
\end{alignat}
\end{subequations}
Constraint \eqref{cons4} shows that the $j$-th \ac{NFP} can establish a maximum of $N_{l_j}$ links with the \acp{SC} as per the number of transceivers. Further, each \ac{SC} can be associated to a maximum of one \ac{NFP} that is included in constraint \eqref{cons5}.

\section{Problem Approximation and Upper Bounds}\label{sec:Prob_Approx}
This section presents the analysis of the association problem \eqref{eq:Opt_Prob}. First of all, we show that if some of the constraints of the problem in \eqref{eq:Opt_Prob} are relaxed, then it exactly maps to the \ac{GAP}. Then, using the relation with \ac{GAP}, it is shown that problem \eqref{eq:Opt_Prob} is at least NP-hard. Further, an analysis for the upper bound of the association problem \eqref{eq:Opt_Prob} without constraints \eqref{cons1} and \eqref{cons4} is presented. Furthermore, we study the \ac{LP} relaxation of the association problem \eqref{eq:Opt_Prob} to obtain another upper bound. In addition, to get a tighter upper bound, a B\&B method is used for the association problem \eqref{eq:Opt_Prob} without neglecting any constraints. Finally, the bounds and relaxed solutions are numerically compared to one another, as well as to those proposed in the next section.

\subsection{Relation to the \ac{GAP}}\label{sec:GAP_relevace}
Here, first of all, we define the \ac{GAP} and then show its relevance with the association problem \eqref{eq:Opt_Prob}. Consider $n$ tasks to be assigned to $m$ agents, where $j$-th agent can complete $i$-th task as per its own capability. This means that the $j$-th agent can complete the $i$-th task with a cost/weight $w_{ij}$ that then returns a utility/profit $p_{ij}$. Thus, the weights and profits are dependent on the $j$-th agent for the $i$-th task. The objective is to maximize the overall profit by assigning each task to exactly one agent without exceeding the capacity of the $j$-th agent, $c_j$. Such a problem is known as \ac{GAP} \cite{martello1990knapsack} and can be written as
\begin{subequations}
\label{eq:GAP}
\begin{alignat}{3}
    \max_{\{x_{ij}\}} \quad &\sum_{i=1}^{n} \sum_{j=1}^{m} p_{ij} \cdot x_{ij} \label{GAP1}\\
    \intertext{subject to}
    &\sum_{i=1}^{n} w_{ij} \cdot x_{ij} \; \leq \; c_j, && \quad \text{for } j \in \Mc, \label{GAP2}\\
    &\sum_{j=1}^{m} x_{ij} \; = \; 1, && \quad \text{for } i \in \Nc, \label{GAP3}\\
    &x_{ij} \; \in \; \left\{ 0, 1 \right\}, && \quad \text{for } i \in \Nc \; \& \; j \in \Mc, \label{GAP4}
\end{alignat}
\end{subequations}
where
\begin{equation}\label{GAP_var}
      x_{ij} =
        \left\{
          \begin{array}{ll}
            1, & \hbox{if $i$-th task is assigned to the $j$-th agent,} \\
            0, & \hbox{otherwise.}
          \end{array}
        \right.
\end{equation}
and the following are the usual restrictions on the \ac{GAP} variables
\begin{align}
  &p_{ij}, \; w_{ij} \; \& \; c_j \text{ are positive integers}, \label{GAP_rest1} \\
  & | \left\{ j: w_{ij} \leq c_j \right\}| \geq 1 \quad \text{for } i \in \Nc, \label{GAP_rest2} \\
  & c_j \geq \min_{i \in \Nc} \left\{ w_{ij} \right\} \qquad \quad \text{for } j \in \Mc. \label{GAP_rest3}
\end{align}
If the weights are fractional, thus condition \eqref{GAP_rest1} is violated, then it can be handled by multiplication of weights with a proper factor. If $i$-th task does not satisfy condition \eqref{GAP_rest2}, then that task cannot be assigned to any agent and \ac{GAP} instance is infeasible. However, another variant of \ac{GAP}, known as LEGAP always admits a feasible solution as in its definition the equality in constraint \eqref{GAP3} is replaced with an inequality, such that $\sum_{j=1}^{m} x_{ij} \leq 1, \forall i \in \Nc$, that allows $i$-th agent to be un-associated under certain conditions \cite{martello1990knapsack}. The agents that violates condition \eqref{GAP_rest3} can be removed from the problem.

Comparison of \ac{GAP} in \eqref{eq:GAP} with the association problem \eqref{eq:Opt_Prob} shows that if the constraints \eqref{cons1}, \eqref{cons3} and \eqref{cons4} are neglected, then the resulting relaxed association problem is equivalent to \ac{GAP}. Note that, the association problem \eqref{eq:Opt_Prob} without constraints \eqref{cons1} and \eqref{cons4} will be referred to as the relaxed association problem. Below, we show the relevance of the relaxed association problem with \ac{GAP} and then discuss the effect of neglecting constraints \eqref{cons1} and \eqref{cons4}. Later on, we also incorporate the constraint \eqref{cons3} in the derivation of the upper bound of the association problem \eqref{eq:Opt_Prob}.

The objective variables $x_{ij}$ in \eqref{GAP1} and $A_{ij}$ in \eqref{Obj_Fun} exactly match each other, where the subscripts $i$ denoting tasks in \eqref{GAP1} is equivalent to \acp{SC} in \eqref{Obj_Fun} and $j$ denoting agents in \eqref{GAP1} is equivalent to \acp{NFP} in \eqref{Obj_Fun}. This means that assigning $i$-th task to $j$-th agent is the same as associating the $i$-th \ac{SC} to the $j$-th \ac{NFP}. The variables $p_{ij}$, $w_{ij}$ and $c_j$ in \eqref{eq:GAP} depend on both tasks and agents, which is the same as of $r_{ij}$, $b_{ij}$ and $B_j$ in association problem \eqref{eq:Opt_Prob} that depend on both \acp{SC} and \acp{NFP}; these variables are equivalent, respectively. The constraint \eqref{cons2} that keeps track of the bandwidth limit of the $j$-th \ac{NFP} is the same as of constraint \eqref{GAP2} that tracks the maximum capacity of $j$-th agent, thus they are equivalent as well. The variables in association problem \eqref{eq:Opt_Prob} satisfy the restrictions of the \ac{GAP} variables given in \eqref{GAP_rest1} to \eqref{GAP_rest3}. For some cases, if the variables in association problem \eqref{eq:Opt_Prob} are in fractions then they can be converted to integers with a multiplication of an appropriate factor.

The above discussion shows that the relaxed association problem \eqref{eq:Opt_Prob} without constraints \eqref{cons1} and \eqref{cons4} is equivalent to \ac{GAP}.

\subsection{Complexity of the Association Problem}\label{sec:NPhard}
The relevance of association problem \eqref{eq:Opt_Prob} with the \ac{GAP} can be exploited to study its complexity.

\begin{prop}
  Association problem \eqref{eq:Opt_Prob} is at least NP-hard.
\end{prop}

\begin{proof}
  As association problem \eqref{eq:Opt_Prob} is a subset of the relaxed association problem, thus to show that the association problem \eqref{eq:Opt_Prob} is at least NP-Hard, it would be enough to show that the relaxed association problem is NP-hard. It is shown in section \ref{sec:GAP_relevace} that the relaxed association problem is equivalent to \ac{GAP}, where \ac{GAP} is an NP-hard problem \cite{fisher1986}. Thus, as \ac{GAP} is NP-hard, so the equivalent relaxed association problem is also NP-hard. This shows that the association problem \eqref{eq:Opt_Prob}, which is a subset of the relaxed association problem, must be at least NP-hard.
\end{proof}

\subsection{Upper Bound of the Optimization Problem}\label{sec:UB}
\begin{remk}\label{lem1}
  Neglecting constraints \eqref{cons1} and \eqref{cons4} of the association problem \eqref{eq:Opt_Prob} results in a relaxed upper bound that is not lower than the original optimal objective function in \eqref{eq:Opt_Prob}.
\end{remk}

\begin{proof}
Knowing that problem \eqref{eq:Opt_Prob} is a maximization problem, enlarging the feasible set by removing constraints can only increase the objective.
\end{proof}


As per Remark \ref{lem1}, the relaxed association problem provides a higher sum data rate than the association problem \eqref{eq:Opt_Prob}. Thus, the upper bound of the relaxed association problem can be regarded as the upper bound for the association problem \eqref{eq:Opt_Prob}. Considering the relevance of our relaxed association problem with the \ac{GAP} as shown in section \ref{sec:GAP_relevace} and using the capacity relaxation procedure adopted in \cite{rossBound}, in the following, we derive an upper bound for the relaxed association problem.

In the relaxed association problem, the bandwidth constraint \eqref{cons2} that is equivalent to the capacity constraint \eqref{GAP2} of \ac{GAP} is relaxed such that
\begin{equation}\label{RelBound1}
      b_{ij} \cdot A_{ij} \; \leq \; B_j, \quad \text{for } i \in \Nc \; \& \; j \in \Mc
\end{equation}
Now, we are relaxing \eqref{cons2}, unlike what is stated in Remark \ref{lem1}. After the above relaxation, the resulting problem has an optimal solution $\hat{\Am}$ that is obtained by determining $j$-th \ac{NFP} for $i$-th \ac{SC} such that
\begin{equation}\label{RelBound2}
  j(i) = \arg \max \left\{ r_{ij}: j \in \Mc, b_{ij} \leq B_j, \; \text{SINR}_{ij} \geq \text{SINR}_{\text{min}} \right\} \quad \text{for } i \in \Nc
\end{equation}
and setting $\hat{A}_{i,j(i)} = 1$ and $\hat{A}_{ij} = 0$ for all $j \in \Mc\backslash\{ j(i) \}$. This results in an upper bound
\begin{equation}\label{RelBound3}
      U_0 = \sum_{i=1}^{N_{SC}} r_{i,j(i)}
\end{equation}
which is then improved as follows. Consider $\Lc_j$ to be the list for the $j$-th \ac{NFP} that consists of \acp{SC} associated with it and $O_j$ to be the overload indicator for the $j$-th \ac{NFP} such that $O_j > 0$ shows that the $j$-th \ac{NFP} has exceeded the bandwidth limit $B_j$ and the list and indicator are defined as
\begin{alignat}{3}
  \Lc_j &= \left\{ i: \hat{A}_{ij} = 1 \right\}, \quad &&\text{for } j \in \Mc \label{RelBound4} \\
  O_j &= \sum_{i \in \Lc_j} b_{ij} - B_j, \quad &&\text{for } j \in \Mc \label{RelBound5}
\end{alignat}
Let us define a set $\acute{\Mc}$ consisting of those \acp{NFP} for which the relaxed constraint \eqref{cons2} is violated and consider $\acute{\Lc}$ to be the list of \acp{SC} associated with those \acp{NFP} that violates constraint \eqref{cons2}. As per the definitions, these sets can be written as
\begin{alignat}{3}
  \acute{\Mc} &= \left\{ j: O_j > 0 \right\} \label{RelBound6} \\
  \acute{\Lc} &= \bigcup_{j \in \acute{\Mc}} \Lc_j \label{RelBound7}
\end{alignat}
If the $i$-th \ac{SC} that is currently associated with the $j$-th \ac{NFP} is reassigned to the other \ac{NFP} such that it results in second maximum data rate, then the resulting minimum penalty is given as
\begin{equation}\label{RelBound8}
  q_i = r_{i,j(i)} - \max \left\{ r_{ij}: j \in \Mc \setminus j(i), b_{ij} \leq B_j, \; \text{SINR}_{ij} \geq \text{SINR}_{\text{min}} \right\}, \quad \text{for } i \in \acute{\Lc}
\end{equation}
This results in a lower bound on the maximum achievable sum data rate in order to satisfy constraint \eqref{cons2}, due to penalty $q_i$. Now, for each \ac{NFP} $j \in \acute{\Mc}$, the objective is to minimize the reassignment penalty by solving the 0-1 single knapsack problem \cite{martello1990knapsack} that can be written as
\begin{subequations}
\label{eq:KP}
\begin{alignat}{3}
    \min v_j = &\sum_{i \in \Lc_j} q_i \; y_{ij} \label{KP1}\\
    \intertext{subject to}
    &\sum_{i \in \Lc_j} b_{ij} \; y_{ij} \geq O_j, \label{KP2}\\
    &y_{ij} \; \in \; \left\{ 0, 1 \right\}, && \quad \text{for } i \in \Lc_j, \label{KP3}
\end{alignat}
\end{subequations}
where $y_{ij} = 1$ if and only if the $i$-th \ac{SC} is dissociated from the $j$-th \ac{NFP}. The resulting bound is thus
\begin{equation}\label{RelBound_Fin}
      U_1 = U_0 - \sum_{j \in \acute{\Mc}} v_j.
\end{equation}
Using the bound $U_1$, we follow a B\&B method presented in \cite{rossBound}, where each branch is bounded by $U_1$. This will provide a solution of the relaxed association problem that can be used as an upper bound of the association problem \eqref{eq:Opt_Prob}.

\subsection{\ac{LP} relaxation and bound}\label{sec:LP_bound}
In this section, we will use the \ac{LP} relaxation on the association problem \eqref{eq:Opt_Prob}. This means that, we will relax the binary constraint on the association matrix $\Am$ as defined in \eqref{cons6}. So, now the entries of the association matrix can vary between 0 and 1 such that $0 \leq A_{ij} \leq 1$. Note that, this will each \ac{SC} to be associated with multiple \acp{NFP}, that will result in relaxation of the constraint \eqref{cons5}. Note that, constraint \eqref{cons5} along with the binary constraint \eqref{cons6} previously restricted each \ac{SC} to be associated to a maximum of one \ac{NFP} only. However, such a relaxation allows us to solve the association problem using convex programming tools. Such a solution can be regarded as a bound of the optimization problem \eqref{eq:Opt_Prob}, which considers all the constraints except the constraint \eqref{cons6}, i.e., binary condition over the association matrix.

\section{Proposed Solution}\label{sec:Prop_Sol}
It is shown in section \ref{sec:NPhard} that the association problem \eqref{eq:Opt_Prob} is at least NP-hard. It is well known that there exists no standard method to solve such an NP-hard problem \cite{NPhard1} and \cite{NPhard2}. Therefore, we have presented two simple bounds of the problem to show the closeness of our proposed greedy solutions with the bounds. Further, we use B\&B method to get the exact solution of the association problem \eqref{eq:Opt_Prob}, where B\&B being an exhaustive search is considered here as a benchmark solution of the association problem \eqref{eq:Opt_Prob}. We call it an exact solution as it does not involve any relaxation of the constraints as compared to the solution bounds presented in sections \ref{sec:UB} and \ref{sec:LP_bound}.

To get an efficient and less computationally-complex solution of the NP-hard association problem \eqref{eq:Opt_Prob}, we present here two greedy solutions that are designed to solve \eqref{eq:Opt_Prob} without relaxing any constraints. One of them is designed for the case where the processing power of \acp{SC} and \acp{NFP} is utilized and is named as Modified Distributed Maximal Demand Minimum Servers (M(DM)$^2$S) algorithm. The other greedy solution is designed for the case of \ac{C-RAN} architecture, where \acp{SC} and \acp{NFP} lack the processing power, and thus, the algorithm runs at the mother-\ac{NFP} or baseband unit (BBU) pool. Therefore, it is named as Centralized Maximal Demand Minimum Servers (CMDMS) algorithm.

Since our scope of work does not focus on finding an optimal 3D placement of \acp{NFP}, the random positions of \acp{SC} and \acp{NFP} are given as a realization of the \emph{Mat\'ern} type-I process as indicated earlier. First of all, we present a system initialization algorithm that requires some of the known system parameters as an input and provide a random distribution of \acp{SC} and \acp{NFP} in a specified rectangular region of area $A_S$. Using this we can generate a number of random scenarios with varying positions of \acp{SC} and \acp{NFP} to evaluate the association algorithms under differen scenarios.

\subsection{System Initialization}\label{sec:Init_Sys}
The algorithm \ref{algo:Intial} mainly consists of two steps. One of which deals with the the \acp{SC}, while the other accounts for the distribution of the \acp{NFP}. For the distribution of the \acp{SC}, as a system parameter, the algorithm needs to know about the rectangular area, $A_S$, average density of \acp{SC}, $\lambda$ per m$^2$, the minimum separation between them, $s_{SC}^{\text{min}}$ in meters, and the number of \acp{SC}, $N_{SC}$. Using these parameters, \acp{SC} are distributed randomly using \emph{Mat\`ern} type-I hard-core process. This provides $N_{SC}^{\text{avg}} = \lambda \exp(-\lambda {s_{SC}^{\text{min}}}^2) A_S$ average number of random distribution points in the rectangular area, $A_S$. Finally, we pick $N_{SC}$ points out of the generated points that provides the 2D locations of $N_{SC}$ \acp{SC} as ($x_i, y_i$), where $i \in \Nc$.

For the \acp{NFP}, we assume symmetry in the case of number of links and bandwidth such that each \ac{NFP} can support a fixed number of links $N_l$ and fixed maximum bandwidth $B$, i.e., $N_{l_j} = N_l$ and $B_j = B$ $\forall$ $j \in \Mc$. The algorithm \ref{algo:Intial} either has the information of the number of \acp{NFP}, $N_D$ as an input parameter or it can be computed as follows. To compute the minimum number of required \acp{NFP}, the algorithm uses the input information $B, b_i$ and $N_l$ and computes the maximum number of \acp{SC} that can be associated with a single \ac{NFP}, $N_{SC}^D$ depending on the bandwidth information such that
\begin{equation}\label{NumDrones1}
  N_{SC}^D = \lfloor \frac{B}{b_{\text{avg}}} \rfloor,
\end{equation}
where $b_{\text{avg}} = \frac{1}{N_{SC}} \sum_{i \in \Nc} b_i$ is the average bandwidth required by a \ac{SC}. Now, the minimum number of required \acp{NFP} is computed as
\begin{equation}\label{NumDrones2}
  N_D = \lceil \frac{N_{SC}}{\min\{N_l,N_{SC}^{D}\}} \rceil.
\end{equation}
Now, again we use another symmetry for the \acp{NFP} with the assumption that the height of every \ac{NFP} is fixed to a maximum defined height $h_{\text{max}}$ such that $h_{D_j} = h_{\text{max}}$, where $h_{\text{max}}$ is obtained by the initialization algorithm as an input parameter. Next, using $h_{\text{max}}$ and considering a fixed maximum path loss $\text{PL}_\text{max}$ in \eqref{PathLoss}, we can compute $s$ that corresponds to the maximum distance covered by $j$-th \ac{NFP}. This maximum distance is equivalent to the required minimum separation between \acp{NFP} $s_D^{\text{min}}$. Finally, we distribute the \acp{NFP} using \emph{Mat\`ern} type-I hard-core process with all the parameters same as of the distribution of \acp{SC} except the separation distance being equal to $s_D^{\text{min}}$. This provides 3D locations of the \acp{NFP} as ($x_{D_j}, y_{D_j}, h_D$). The above mentioned procedure is summarized in Algorithm \ref{algo:Intial}.

\begin{varalgorithm}{Initialization}
\renewcommand{\baselinestretch}{1.1}
\small
\caption{System Initialization}
\label{algo:Intial}
\begin{algorithmic}[1]
    \REQUIRE $\lambda, \; A_S, \; s_{BS}^{\text{min}}, \; h_{\text{max}}, \; \text{PL}_{\text{max}}, \; \alpha, \; \beta, \; \eta_{\text{LoS}}, \; \eta_{\text{NLoS}}, N_{SC}, \; N_D$
    \ENSURE $\left(x_i,\, y_i\right), \; (x_{D_j}, y_{D_j}, h_D)$
    \STATE \textbf{Distribution of \acp{SC}:}
    \STATE $(x_i, y_i) \leftarrow$ \emph{Mat\`ern} Process$(A_S, \lambda, s_{SC}^{\text{min}})$
    \STATE $(x_i, y_i) \leftarrow$ Randomly select $N_{SC}$ number of points out of $(x_i, y_i)$
    \STATE \textbf{Distribution of \acp{NFP}:}
    \STATE Compute $s_D^{\text{min}}$ using \eqref{Prob_LoS}, \eqref{PathLoss}, $PL_{\text{max}}, \alpha, \beta, \eta_{\text{LoS}}, \eta_{\text{NLoS}}$
    \IF {$N_D = 0 $}
        \STATE Compute $N_D$ using \eqref{NumDrones2}
    \ENDIF
    \STATE $(x_{D_j}, y_{D_j}) \leftarrow$ \emph{Mat\`ern} Process$(A_S, \lambda, s_D^{\text{min}})$
    \STATE $(x_{D_j}, y_{D_j}) \leftarrow N_D$ points out of $(x_{D_j}, y_{D_j})$ and $h_{D_j} = h_{\text{max}}$
\end{algorithmic}
\end{varalgorithm}

Next, the SINR parameter \eqref{eq:SINR} for each pair of \ac{SC} and \ac{NFP} is computed using a snapshot of the above distribution of \acp{SC} and \acp{NFP} combined with the bandwidth and data rate requirements of \acp{SC}, i.e., $b_{ij}$ and $r_{ij}$, respectively. All this information is then passed to the below presented algorithms to find the association of \acp{SC} and \acp{NFP} by solving the association problem \eqref{eq:Opt_Prob}.

\subsection{Modified Distributed Maximal Demand Minimum Servers Algorithm}\label{sec:DMDMSA}
This algorithm is designed to use the processing power of three network nodes including \acp{SC}, \acp{NFP} and mother-\ac{NFP}. Therefore, it is divided into three steps that are distributed among those three nodes, i.e., first step at \acp{SC}, second step at \acp{NFP} and third step at mother-\ac{NFP}. Mainly, the second step that is divided among \acp{NFP} speeds up the optimization process. Each step takes care of one or more constraints of the association problem \eqref{eq:Opt_Prob}.

\subsubsection{\textbf{Step 1}}\label{sec:Step1Sol}
This step is performed at the $i$-th \ac{SC} individually. The $i$-th \ac{SC} uses SINR parameter from \eqref{eq:SINR} and compares it with the minimum SINR requirement $\text{SINR}_\text{min}$ satisfying constraint \eqref{cons3}. This provides a list of possible pairs with \acp{NFP} for the $i$-th \ac{SC}, $i \in \Nc$. Out of its list, the $i$-th \ac{SC} picks the $j$-th \ac{NFP}, $j \in \Mc$ that results in maximum value of the decision ratio $r_{ij}/b_{ij}$ and sends the association request to only the selected \ac{NFP}. As each \ac{SC} selects only one \ac{NFP}, this procedure takes care of the constraint \eqref{cons5}. Note that the decision ratio is designed keeping in view the objective function \eqref{Obj_Fun} and constraint \eqref{cons2}, where we want to maximize the data rate $r_{ij}$ and minimize the bandwidth $b_{ij}$\footnote{This decision ratio is inspired by the optimal solution for the Knapsack problem, which has similarities to our problem.}.

\subsubsection{\textbf{Step 2}}\label{sec:Step2Sol}
At this step, the $j$-th \ac{NFP} uses information about its maximum number of supported links and bandwidth, and initializes counters $C_{N_l}^j = 0$ and $C_b^j = 0$. The $j$-th \ac{NFP}, $j \in \Mc$ receives a number of association requests from a group of \acp{SC}. The $j$-th \ac{NFP} goes through its own list of association requests and selects, one by one in an ascending order, the \acp{SC} resulting in the maximum decision ratio $r_{ij}/b_{ij}$ till the end of the list, or till it reaches its bandwidth/links capacity, whichever occurs earlier. Before associating the selected request of the $i$-th \ac{SC}. The $j$-th \ac{NFP} verifies the constraints for the maximum number of supported links $N_l$ and maximum bandwidth $B$ as follows. \ac{NFP} verifies if it has remaining resources to serve the selected request of \ac{SC} i.e,. $C_{N_l}^j < N_l$ and $C_b^j + b_{ij} \leq B$. If those two conditions are satisfied then it associates the $i$-th \ac{SC} and updates its respective association entry $A_{ij}$ and related counters as $C_{N_l} = C_{N_l} + 1$ and $C_b^j = C_b^j + b_{ij}$. If any of the two limits including the number of links and the bandwidth for the $j$-th \ac{NFP} is reached, i.e., $C_{N_l}^j \nless N_l$ and $C_b^j \nless B$, then the association process for this \ac{NFP} ends. Thus, at this step, the $j$-th \ac{NFP} takes care of the two constraints including the maximum number of links, i.e., \eqref{cons4} and the maximum bandwidth, i.e., \eqref{cons2}. Furthermore, in case if all the requests of \acp{SC} are entertained already and no further request is remaining for the $j$-th \ac{NFP}, then the process at this step completes for the $j$-th \ac{NFP}.

This step is designed to use the processing power of \acp{NFP} and further it is distributed among them in such a way so it can be performed in parallel. This distribution and parallel processing speeds up the overall association process. Also, note that until this step, we have satisfied constraints \eqref{cons2} to \eqref{cons6} only. We have also used the information of constraint \eqref{cons1} but has not verified it yet, as all the information is distributed between \acp{NFP} and \acp{SC}, so there is no way to collectively keep track of the combined data rate information.

\subsubsection{\textbf{Step 3}}\label{sec:Step3Sol}
All of the information at step 2 is shared with the mother-\ac{NFP}. Mother-\ac{NFP} generates the association matrix $\Am$, where $j$-th column has $N_l$ number of ones at maximum. It initializes the data rate counter $C_r$ with the currently assigned total data rate of the associated \acp{SC} and keep track of the sum data rate as per the association matrix. Thus, at this step, mother-\ac{NFP} verifies the constraint \eqref{cons1} as follows.

If the backhaul data rate limit is not reached yet, i.e., $C_r < R$, then mother-\ac{NFP} goes through the association matrix to look for the \acp{SC} not associated with any \ac{NFP}. For those remaining \acp{SC}, mother-\ac{NFP} creates a list of possible \ac{NFP}-\ac{SC} pairs. Out of the list, the \ac{NFP}-\ac{SC} pairs not satisfying the SINR constraint \eqref{cons3} are discarded out of the list. Then, mother-\ac{NFP} selects the \ac{NFP}-\ac{SC} pair with maximum decision ratio $r_{ij}/b_{ij}$. For the selected $j$-th \ac{NFP}, mother-\ac{NFP} first verifies the number of links and bandwidth resources, i.e., $C_{N_l}^j < N_l$ and $C_b^j < B$. In case either of the number of links or bandwidth limits have been reached, all the respective \ac{NFP}-\ac{SC} pairs of the $j$-th \ac{NFP} are discarded from the list. Otherwise, the backhaul data rate and bandwith constraints, i.e., \eqref{cons1} and \eqref{cons2}, are verified for the selected \ac{NFP}-\ac{SC} pair such that $C_r + r_{ij} \leq R$ and $C_b^j + b_{ij} \leq B$. If the constraints are verified then mother-\ac{NFP} associates the \ac{NFP}-\ac{SC} pair and updates the association matrix $\Am$, and the related counters such that $C_{N_l}^j = C_{N_l}^j + 1$, $C_b^j = C_b^j + b_{ij}$ and $C_r = C_r + r_{ij}$. Also, for the selected \ac{SC} associated with the $k$-th \ac{NFP}, other possible links are neglected, i.e., links with the $j$-th \acp{NFP}, $j \in \Mc$, where $j \neq k$ are discarded from the list to satisfy constraint \eqref{cons5}. Then, mother-\ac{NFP} selects the next \ac{NFP}-\ac{SC} pair resulting in next maximum decision ratio and keeps associating the remaining \acp{SC} until the resources are fully utilized or all of the \acp{SC} gets associated.

The other case that needs to be checked is when the backhaul data rate limit has been exceeded, i.e., $C_r > R$. This may happen due to the distributed nature of the algorithm, as until step 2 there is no centralized tracking of the sum data rate for all of the \ac{NFP}-\ac{SC} associations. For this case, mother-\ac{NFP} goes through the association matrix $\Am$ to find out the associated \ac{NFP}-\ac{SC} pairs. Out of those pairs, it selects the one that results in minimum value of data rate $r_{ij}$. Then, it disassociates the selected pair and updates the association matrix index $A_{ij}$ and the counters as $C_{N_l}^j = C_{N_l}^j - 1$, $C_b^j = C_b^j - b_{ij}$ and $C_r = C_r - r_{ij}$. Same procedure is followed until the backhaul data rate limit is satisfied.

Throughout this algorithm, priority is given to the \ac{NFP}-\ac{SC} pairs resulting in maximum decision ratio which means that the algorithm is designed to increase the data rate under the bandwidth limit mainly. This is in accordance with the objective function \eqref{Obj_Fun} and thus this algorithm focuses on user-centric case where \acp{SC} with users who demand high data rate are given priority. This algorithm provides an efficient solution in three simple steps with less computational complexity as compared to B\&B method and is summarized in Algorithm \ref{algo:DMDMSA}. Note that, we had presented a similar algorithm named Distributed Maximal Demand Minimum Servers ((DM)$^2$S) in \cite{AwaisGlobeCom}, where a different decision ratio was used and the case of $C_r < R$ at step 3 was not considered.

\begin{varalgorithm}{M(DM)$^2$S}
\renewcommand{\baselinestretch}{1.1}
\small
\caption{Modified Distributed Maximal Demand Minimum Servers Algorithm}
\label{algo:DMDMSA}
\begin{algorithmic}[1]
    \REQUIRE $N_{SC}, \; N_D, \; \text{SINR}_{\text{max}}, \; N_l, \; B, \; R, \; \text{SINR}_{ij}, \; r_{ij}, \; b_{ij}$
    \ENSURE $\Am$
    \STATE Initialize: $\Am = \boldsymbol{0}_{N_{SC} \times N_D}$
    \STATE \textbf{Step 1:} at each \ac{SC}
    \FOR {$i = 1$ \TO $N_{SC}$}
        \STATE Create a list of \acp{NFP} satisfying $\text{SINR}_{ij} \geq \text{SINR}_{\text{min}}$
        \STATE Out of the list, select and request $j$-th \ac{NFP} with max. $r_{ij}/b_{ij}$
    \ENDFOR
    \STATE \textbf{Step 2:} at each \ac{NFP}
    \STATE Initialize counters: $C_{N_l}^j = 0, \; C_b^j = 0$ for $j \in \Mc$
    \FOR {$j = 1$ \TO $N_D$}
        \WHILE {$C_{N_l}^j < N_l$ AND $C_b^j < B$ AND List for the $j$-th \ac{NFP} not empty}
            \STATE Select $i$-th \ac{SC} with max. $r_{ij}/b_{ij}$
            \IF {$C_b^j + b_{ij} \leq B$}
                \STATE Update $A_{ij}=1$, $C_{N_l}^j = C_{N_l}^j+1$ and $C_b^j = C_b^j + b_{ij}$
            \ENDIF
        \ENDWHILE
    \ENDFOR
    \STATE \textbf{Step 3:} at mother-\ac{NFP}
    \STATE Initialize: $C_r$ as total data rate of associated \acp{SC}
    \STATE Create a list of un-associated \acp{SC} by scanning $\Am$
    \STATE Initialize counters $C_{N_l}^j$ and $C_b^j$ by scanning $\Am$
    \WHILE {$C_r < R$ AND List of un-associated \acp{SC} not empty}
        \STATE Out of the list, discard \ac{NFP}-\ac{SC} pairs with $\text{SINR}_{ij} < \text{SINR}_{\text{min}}$
        \STATE Select \ac{NFP}-\ac{SC} pair with max. $r_{ij}/b_{ij}$
        \IF {$C_{N_l}^j < N_l$ AND $C_b^j < B$}
            \IF {$C_r + r_{ij} \leq R$ AND $C_b^j + b_{ij} \leq B$}
                \STATE Update $A_{ij}=1$, $C_{N_l}^j = C_{N_l}^j+1$, $C_r = C_r + r_{ij}$ and $C_b^j = C_b^j + b_{ij}$
                \STATE Discard other pairs of the $i$-th \ac{SC} from the list
            \ENDIF
        \ELSE
            \STATE Discard all pairs of the $j$-th \ac{NFP} from the list
        \ENDIF
    \ENDWHILE
    \WHILE {$C_r > R$}
        \STATE Create a list of associated \ac{NFP}-\ac{SC} pairs
        \STATE Select the pair with min. $r_{ij}$
        \STATE De-associate the selected pair and update $A_{ij}=0$, $C_{N_l}^j = C_{N_l}^j-1$, $C_r = C_r - r_{ij}$ and $C_b^j = C_b^j - b_{ij}$
    \ENDWHILE
\end{algorithmic}
\end{varalgorithm}

\subsection{Centralized Maximal Demand Minimum Servers Algorithm}\label{sec:CMDMSA}
This algorithm is designed for the \ac{C-RAN} architecture where the processing takes place mainly at the baseband unit (BBU) pool. Thus, here we consider that the \acp{SC} and \acp{NFP} only carry the control information and all the data processing takes place at either the mother-\ac{NFP} and the BBU pool. Both the mother-\ac{NFP} or he BBU pool receive all the necessary information from the \acp{SC} and \acp{NFP}. Similar to the distributed algorithm, this one is designed also for the user-centric case where priority is given to the \acp{SC} demanding a high data rate and keeping in view the bandwidth constraint. Thus, we use the same decision ratio in this algorithm as used in Algorithm \ref{algo:DMDMSA}.

Mother-\ac{NFP} receives the necessary information about the \acp{SC} and \acp{NFP} such as SINR of the \ac{NFP}-\ac{SC} links $\text{SINR}_{ij}$, minimum SINR requirement of the system $\text{SINR}_{\text{min}}$, demanded bandwidth $b_{ij}$ and data rate $r_{ij}$ of \acp{SC}, number of links $N_{l_j}$ and bandwidth $B$ limits of the \acp{NFP} and backhaul data rate limit $R$. Using all of the above control information, mother-\ac{NFP} creates a list of \ac{NFP}-\ac{SC} pairs that satisfy the SINR constraint \eqref{cons3}. It also initializes the counters form zero for the number of links $C_{N_l}^j$ assigned to the $j$-th \ac{NFP}, assigned bandwidth $C_b^j$ to the $j$-th \ac{NFP} and assigned sum data rate $C_r$ of all the \acp{NFP}. The backhaul data rate limit $R$, i.e., constraint \eqref{cons1} is verified by mother-\ac{NFP} such as $C_r < R$. If the verification fails, the algorithm terminates. Otherwise, the association proceeds as follows. Mother-\ac{NFP} goes through the list of \ac{NFP}-\ac{SC} pairs and selects the pair that provides maximum decision ratio $r_{ij}/b_{ij}$. For the selected $j$-th \ac{NFP}, it verifies the number of links $N_l$ and bandwidth $B$ limits such as $C_{N_l}^j < N_l$ and $C_b^j < B$. If any one of the two limits is exceeded, it means the selected \ac{NFP} cannot provide further resources for the association of \acp{SC}. Thus, the selected and remaining pairs related to the selected $j$-th \ac{NFP} are discarded from the list. If the two constraints are satisfied, then mother-\ac{NFP} verifies the sum data rate and bandwidth constraints \eqref{cons1} and \eqref{cons2}, respectively, using demanded data rate and bandwidth information of the selected \ac{NFP}-\ac{SC} pair such as $C_r + r_{ij} \leq R$ and $C_b^j + b_{ij} \leq B$. If the constraints are satisfied the mother-\ac{NFP} associates the selected pair by modifying the association matrix entry as $A_{ij} = 1$ and updates the respective counters accordingly. Further, after the association, the remaining links of the $i$-th \ac{SC} with other \acp{NFP} except the selected \ac{NFP} are discarded from the list. In case the constraints are not satisfied, the selected pair is discarded from the list. The algorithm terminates if either the sum data rate limit is reached or the list of \ac{NFP}-\ac{SC} pairs to be associated ends. The whole procedure is summarized in Algorithm \ref{algo:CMDMSA}.

\begin{varalgorithm}{CMDMSA}
\renewcommand{\baselinestretch}{1.1}
\small
\caption{\small{Centralized Maximal Demand Minimum Servers Algorithm}}
\label{algo:CMDMSA}
\begin{algorithmic}[1]
    \REQUIRE $N_{SC}, \; N_D, \; N_l, \; \text{SINR}_{\text{min}}, \; B, \; R, \; \text{SINR}_{ij}, \; r_{ij}, \; b_{ij}$
    \ENSURE $\Am$
    \STATE Create a list of \ac{NFP}-\ac{SC} pairs satisfying $\text{SINR}_{ij} \geq \text{SINR}_{\text{min}}$
    \STATE Initialize counters: $C_{N_l}^j = 0$, $C_b^j = 0$ for $j \in \Mc$ and $C_r = 0$
    \WHILE {List of \ac{NFP}-\ac{SC} pairs is not empty AND $C_r < R$}
        \STATE Select \ac{NFP}-\ac{SC} pair with max. $r_{ij}/b_{ij}$
        \IF {$C_{N_l}^j < N_l$ AND $C_b^j < B$}
            \IF {$C_r + r_{ij} \leq R$ AND $C_b^j + b_{ij} \leq B$}
                \STATE Update $A_{ij}=1$, $C_{N_l}^j = C_{N_l}^j+1$, $C_r = C_r + r_{ij}$, and $C_b^j = C_b^j + b_{ij}$
                \STATE Discard other pairs of \ac{SC} from the list
            \ELSE
                \STATE Discard selected \ac{NFP}-\ac{SC} pair from the list
            \ENDIF
        \ELSE
            \STATE Discard all pairs of the $j$-th \ac{NFP} from the list
        \ENDIF
    \ENDWHILE
\end{algorithmic}
\end{varalgorithm}

\begin{table}[t!]
\renewcommand{\arraystretch}{1.1}
\centering
\caption{Simulation Parameters}
\label{tab:SimPar}
\begin{tabular}{|c|c|c|c|}
\hline
\textbf{Parameter}          & \textbf{Value}                & \textbf{Parameter}       & \textbf{Value}   \\ \hline
$\alpha$                    & 9.61                          & $\beta$                  & 0.16             \\ \hline
$\eta_{\text{LoS}}$         & 1 dB                          & $\eta_{\text{NLoS}}$     & 20 dB            \\ \hline
$f_c$                       & 2 GHz                         & $P_t$                    & 5 Watts          \\ \hline
$\text{SINR}_{\text{min}}$  & -5 dB                         & $\text{PL}_{\text{max}}$ & 115 dB           \\ \hline
$\lambda$                   & 5 $\times 10^{-6}$ m$^{-2}$   & $h_{D_{\text{max}}}$     & 300 meters       \\ \hline
$N_{SC}$                    & 30                            & $N_D$                    & 3                \\ \hline
$\rv_{SC}$                  & \multicolumn{3}{c|}{ \{ 30, 60, 90, 120, 150 \} Mbps }                      \\ \hline
\end{tabular}
\end{table}

\section{Numerical Results}\label{sec:Sim_Res}

Consider a system as shown in Fig. \ref{fig:SysMod}, where we analyze the association problem of \acp{SC} and \acp{NFP} distributed over a square region of area $A_S=16$ km$^2$. For the distribution of \acp{SC} and \acp{NFP}, we use Algorithm \ref{algo:Intial}. The same average density $\lambda$ is used in \emph{Mat\`ern} type-I hard-core process for the distribution of both \acp{SC} and \acp{NFP}. Neighbouring \acp{SC} are separated by maintaining a minimum separation of $s_{SC}^{\text{min}} = 300$ meters. For the distribution of \acp{NFP}, we compute the minimum separation $s_D^{\text{min}}$ between neighbouring \acp{NFP} using maximum path loss $\text{PL}_{\text{max}}$ as shown in Table \ref{tab:SimPar}. Next, we assign the data rate to the \acp{SC} randomly from a pre-defined vector $\rv_{SC}$, where it is assumed that the $i$-th \ac{SC} will demand same data rate from any of the $j$-th \ac{NFP}, $j \in \Mc$ such that $r_{ij} = r_i$, $\forall$ $j \in \Mc$. Then, using the parameters defined in Table \ref{tab:SimPar}, the demanded bandwidth $b_{ij}$ of the \acp{SC} and $\text{SINR}_{ij}$ of the \ac{NFP}-\ac{SC} pairs are computed. A snapshot is taken of the current scenario and all the relevant information is passed to the algorithms to find the best possible association of \acp{SC} and \acp{NFP}. We consider a number of scenarios with varying distributions of \acp{NFP} and \acp{SC} to analyze the performance of association algorithms over different scenarios.

\begin{figure}[t!]\centering
	\setlength\figureheight{5.5cm}
	\setlength\figurewidth{5.5cm}
    \subfloat[\ac{GAP} Bound association.]{\label{fig:fig1a}
	\footnotesize
\begin{tikzpicture}
\begin{axis}[%
width=\figurewidth,
height=\figureheight,
xmin=0, xmax=4000,
xlabel near ticks,
xlabel={X-coordinate (meters)},
xmajorgrids,
ymin=0, ymax=4000,
ylabel near ticks,
ylabel={Y-coordinate (meters)},
ymajorgrids,
zmin=0, zmax=400,
ztick={\empty},
zmajorgrids,
view={0}{90},
axis x line*=bottom,
axis y line*=left,
axis z line*=left,
]
\addplot[only marks,mark=o,mark options={},mark size=2.3pt,color=black] plot table[row sep=crcr,]{%
3696.82448649415	3833.33081864512\\
1612.07287554816	2420.25400646298\\
928.266860686699	1902.48814417862\\
};

\addplot[only marks,mark=*,mark options={},mark size=2.3pt,color=red] plot table[row sep=crcr,]{%
801.328006848771	1258.98120102693\\
2976.21757114201	172.481934282448\\
1250.22522755541	1709.19199995294\\
1744.02864613869	609.791410700445\\
2673.5181395338	244.430311198728\\
2764.50250505665	596.394148300725\\
1900.81593743129	164.484381489104\\
1784.80876641252	1937.52461384197\\
1593.6730754667	1010.73983918613\\
1375.52348584576	2059.47141956117\\
2376.7900145515	301.613921421186\\
2435.18419678485	749.188482160487\\
};

\addplot3[only marks,mark=asterisk,mark options={},mark size=5pt,color=red] plot table[row sep=crcr,]{%
2140.38242564172	856.976980491098	300\\
};

\addplot[only marks,mark=square*,mark options={},mark size=2pt,color=green] plot table[row sep=crcr,]{%
965.953629110736	3546.3851367558\\
2144.00366039004	3848.42669307757\\
586.933304609412	1886.75812299645\\
293.820976001316	3762.21832903258\\
1968.55335566752	2964.71028035246\\
1262.59644066917	2903.74724075916\\
};
\addplot3[only marks,mark=asterisk,mark options={},mark size=5pt,color=green] plot table[row sep=crcr,]{%
570.820669025533	3822.066610263	300\\
};
\addplot[only marks,mark=diamond*,mark options={},mark size=3.6092pt,color=blue] plot table[row sep=crcr,]{%
3524.41958014915	1018.06381639867\\
3685.10236414622	189.057348689127\\
3949.38501358752	3096.66512174053\\
3758.25936947533	784.278266192113\\
3429.51773356311	2013.34956680105\\
3719.20361848455	2857.12089297272\\
3149.15134313435	1238.82972606779\\
3377.34226356926	2879.98468001952\\
2957.48167473541	2394.42783355965\\
};
\addplot3[only marks,mark=asterisk,mark options={},mark size=5pt,color=blue] plot table[row sep=crcr,]{%
3233.28455918795	1602.91973570834	300\\
};
\end{axis}
\end{tikzpicture}%
} \hspace{1cm}
    \subfloat[LP association.]{\label{fig:fig1b}
	\footnotesize
\begin{tikzpicture}
\begin{axis}[%
width=\figurewidth,
height=\figureheight,
xmin=0, xmax=4000,
xlabel near ticks,
xlabel={X-coordinate (meters)},
xmajorgrids,
ymin=0, ymax=4000,
ylabel near ticks,
ylabel={Y-coordinate (meters)},
ymajorgrids,
zmin=0, zmax=400,
ztick={\empty},
zmajorgrids,
view={0}{90},
axis x line*=bottom,
axis y line*=left,
axis z line*=left
]
\addplot[only marks,mark=o,mark options={},mark size=2.3pt,color=black] plot table[row sep=crcr,]{%
801.328006848771	1258.98120102693\\
1784.80876641252	1937.52461384197\\
};
\addplot[only marks,mark=*,mark options={},mark size=2.3pt,color=red] plot table[row sep=crcr,]{%
2976.21757114201	172.481934282448\\
1250.22522755541	1709.19199995294\\
1744.02864613869	609.791410700445\\
2673.5181395338	244.430311198728\\
2764.50250505665	596.394148300725\\
1900.81593743129	164.484381489104\\
1593.6730754667	1010.73983918613\\
1375.52348584576	2059.47141956117\\
2376.7900145515	301.613921421186\\
2435.18419678485	749.188482160487\\
928.266860686699	1902.48814417862\\
};
\addplot3[only marks,mark=asterisk,mark options={},mark size=5pt,color=red] plot table[row sep=crcr,]{%
2140.38242564172	856.976980491098	300\\
};
\addplot[only marks,mark=square*,mark options={},mark size=2pt,color=green] plot table[row sep=crcr,]{%
965.953629110736	3546.3851367558\\
2144.00366039004	3848.42669307757\\
1612.07287554816	2420.25400646298\\
586.933304609412	1886.75812299645\\
293.820976001316	3762.21832903258\\
1968.55335566752	2964.71028035246\\
1262.59644066917	2903.74724075916\\
};
\addplot3[only marks,mark=asterisk,mark options={},mark size=5pt,color=green] plot table[row sep=crcr,]{%
570.820669025533	3822.066610263	300\\
};
\addplot[only marks,mark=diamond*,mark options={},mark size=3.6092pt,color=blue] plot table[row sep=crcr,]{%
3524.41958014915	1018.06381639867\\
3685.10236414622	189.057348689127\\
3949.38501358752	3096.66512174053\\
3758.25936947533	784.278266192113\\
3696.82448649415	3833.33081864512\\
3429.51773356311	2013.34956680105\\
3719.20361848455	2857.12089297272\\
3149.15134313435	1238.82972606779\\
3377.34226356926	2879.98468001952\\
2957.48167473541	2394.42783355965\\
};
\addplot3[only marks,mark=asterisk,mark options={},mark size=5pt,color=blue] plot table[row sep=crcr,]{%
3233.28455918795	1602.91973570834	300\\
};
\end{axis}
\end{tikzpicture}%
} \hspace{1cm}
    \subfloat[B\&B association.]{\label{fig:fig1c}
	\footnotesize
\begin{tikzpicture}
\begin{axis}[%
width=\figurewidth,
height=\figureheight,
xmin=0, xmax=4000,
xlabel near ticks,
xlabel={X-coordinate (meters)},
xmajorgrids,
ymin=0, ymax=4000,
ylabel near ticks,
ylabel={Y-coordinate (meters)},
ymajorgrids,
zmin=0, zmax=400,
ztick={\empty},
zmajorgrids,
view={0}{90},
axis x line*=bottom,
axis y line*=left,
axis z line*=left,
]
\addplot[only marks,mark=o,mark options={},mark size=2.3pt,color=black] plot table[row sep=crcr,]{%
801.328006848771	1258.98120102693\\
3696.82448649415	3833.33081864512\\
1612.07287554816	2420.25400646298\\
1784.80876641252	1937.52461384197\\
928.266860686699	1902.48814417862\\
};\label{entry_1}

\addplot[only marks,mark=*,mark options={},mark size=2.3pt,color=red] plot table[row sep=crcr,]{%
2976.21757114201	172.481934282448\\
1250.22522755541	1709.19199995294\\
1744.02864613869	609.791410700445\\
2673.5181395338	244.430311198728\\
2764.50250505665	596.394148300725\\
1900.81593743129	164.484381489104\\
1593.6730754667	1010.73983918613\\
1375.52348584576	2059.47141956117\\
2376.7900145515	301.613921421186\\
2435.18419678485	749.188482160487\\
};\label{entry_2}

\addplot3[only marks,mark=asterisk,mark options={},mark size=5pt,color=red] plot table[row sep=crcr,]{%
2140.38242564172	856.976980491098	300\\
};\label{entry_3}

\addplot[only marks,mark=square*,mark options={},mark size=2pt,color=green] plot table[row sep=crcr,]{%
965.953629110736	3546.3851367558\\
2144.00366039004	3848.42669307757\\
586.933304609412	1886.75812299645\\
293.820976001316	3762.21832903258\\
1968.55335566752	2964.71028035246\\
1262.59644066917	2903.74724075916\\
};
\addplot3[only marks,mark=asterisk,mark options={},mark size=5pt,color=green] plot table[row sep=crcr,]{%
570.820669025533	3822.066610263	300\\
};
\addplot[only marks,mark=diamond*,mark options={},mark size=3.6092pt,color=blue] plot table[row sep=crcr,]{%
3524.41958014915	1018.06381639867\\
3685.10236414622	189.057348689127\\
3949.38501358752	3096.66512174053\\
3758.25936947533	784.278266192113\\
3429.51773356311	2013.34956680105\\
3719.20361848455	2857.12089297272\\
3149.15134313435	1238.82972606779\\
3377.34226356926	2879.98468001952\\
2957.48167473541	2394.42783355965\\
};
\addplot3[only marks,mark=asterisk,mark options={},mark size=5pt,color=blue] plot table[row sep=crcr,]{%
3233.28455918795	1602.91973570834	300\\
};
\end{axis}
\end{tikzpicture}%
} \hspace{1cm}
    \subfloat[\ref{algo:DMDMSA} \& \ref{algo:CMDMSA} association.]{\label{fig:fig1d}
	\footnotesize
\begin{tikzpicture}
\begin{axis}[%
width=\figurewidth,
height=\figureheight,
xmin=0, xmax=4000,
xlabel near ticks,
xlabel={X-coordinate (meters)},
xmajorgrids,
ymin=0, ymax=4000,
ylabel near ticks,
ylabel={Y-coordinate (meters)},
ymajorgrids,
zmin=0, zmax=400,
ztick={\empty},
zmajorgrids,
view={0}{90},
axis x line*=bottom,
axis y line*=left,
axis z line*=left,
]
\addplot[only marks,mark=o,mark options={},mark size=2.3pt,color=black] plot table[row sep=crcr,]{%
3696.82448649415	3833.33081864512\\
1612.07287554816	2420.25400646298\\
1784.80876641252	1937.52461384197\\
1375.52348584576	2059.47141956117\\
928.266860686699	1902.48814417862\\
};

\addplot[only marks,mark=*,mark options={},mark size=2.3pt,color=red] plot table[row sep=crcr,]{%
801.328006848771	1258.98120102693\\
2976.21757114201	172.481934282448\\
1250.22522755541	1709.19199995294\\
1744.02864613869	609.791410700445\\
2673.5181395338	244.430311198728\\
2764.50250505665	596.394148300725\\
1900.81593743129	164.484381489104\\
1593.6730754667	1010.73983918613\\
2376.7900145515	301.613921421186\\
2435.18419678485	749.188482160487\\
};

\addplot3[only marks,mark=asterisk,mark options={},mark size=5pt,color=red] plot table[row sep=crcr,]{%
2140.38242564172	856.976980491098	300\\
};

\addplot[only marks,mark=square*,mark options={},mark size=2pt,color=green] plot table[row sep=crcr,]{%
965.953629110736	3546.3851367558\\
2144.00366039004	3848.42669307757\\
586.933304609412	1886.75812299645\\
293.820976001316	3762.21832903258\\
1968.55335566752	2964.71028035246\\
1262.59644066917	2903.74724075916\\
};
\addplot3[only marks,mark=asterisk,mark options={},mark size=5pt,color=green] plot table[row sep=crcr,]{%
570.820669025533	3822.066610263	300\\
};
\addplot[only marks,mark=diamond*,mark options={},mark size=3.6092pt,color=blue] plot table[row sep=crcr,]{%
3524.41958014915	1018.06381639867\\
3685.10236414622	189.057348689127\\
3949.38501358752	3096.66512174053\\
3758.25936947533	784.278266192113\\
3429.51773356311	2013.34956680105\\
3719.20361848455	2857.12089297272\\
3149.15134313435	1238.82972606779\\
3377.34226356926	2879.98468001952\\
2957.48167473541	2394.42783355965\\
};
\addplot3[only marks,mark=asterisk,mark options={},mark size=5pt,color=blue] plot table[row sep=crcr,]{%
3233.28455918795	1602.91973570834	300\\
};
\end{axis}
\end{tikzpicture}%
} \hspace{1cm}
    \subfloat{\label{fig:fig1e}
	\begin{tikzpicture}[auto, node distance=1cm,>=latex',scale=0.65, transform shape,every text node part/.style={align=center}]

\node [draw,fill=white,anchor=north west] at (rel axis cs: 0,0) {\shortstack[l]{
        \ref{entry_1} \, Unassociated SCs ~~
        \ref{entry_2} \, Associated SCs ~~
        \ref{entry_3} \, NFPs}};

\end{tikzpicture}%
}
	\caption{2D view of a spatial pattern and the resulting association of \acp{NFP} and \acp{SC} with constraints $N_l = 10$, $B = 0.4$ GHz, and $R = 2.3$ Gbps.}
	\label{fig:fig1}
\end{figure}
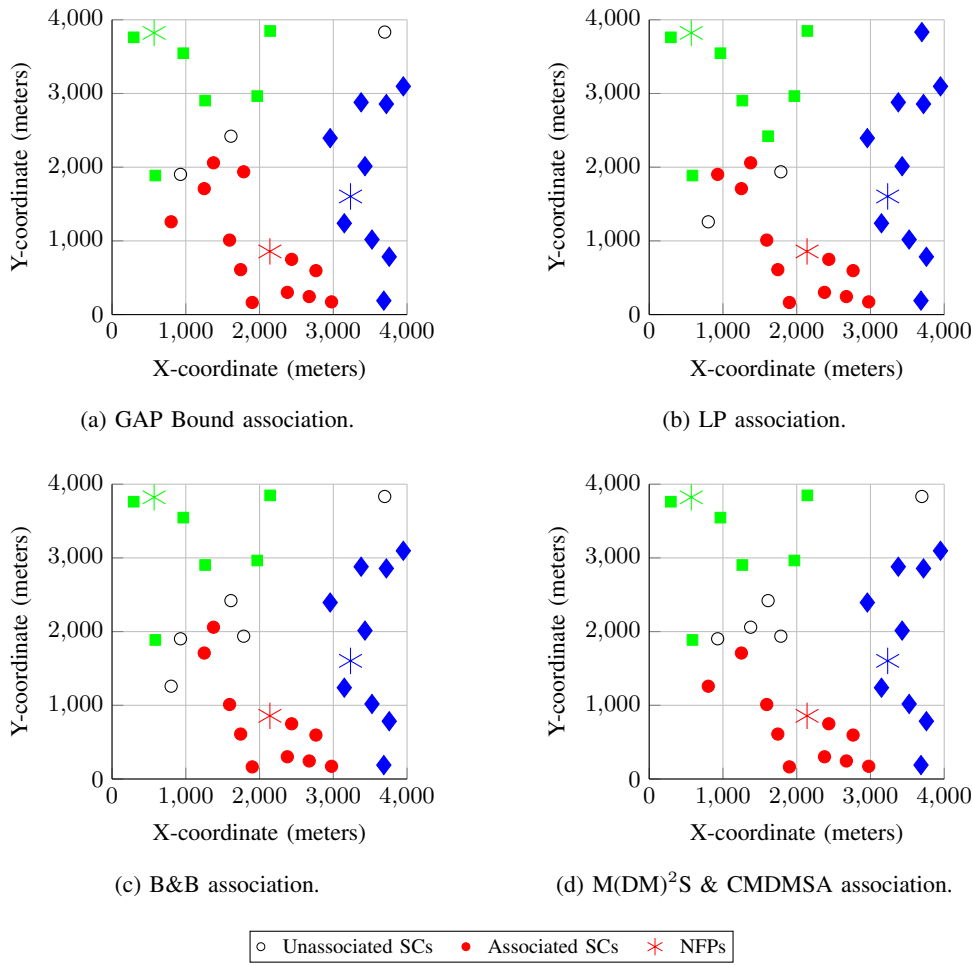

Fig. \ref{fig:fig1} depicts one of the considered scenarios for the distribution of \acp{SC} and \acp{NFP}. Here, only 2D view is shown as the $N_D$ \acp{NFP} are placed at the same height. Fig. \ref{fig:fig1a}  and \ref{fig:fig1b} show the association computed using \ac{GAP} bound and \ac{LP} relaxation, respectively. Further, the association computed using B\&B exhaustive search that considers all the constraints without any relaxation is shown in Fig. \ref{fig:fig1c}. As the association of our greedy algorithms is same for the considered case, so the association is shown jointly in \ref{fig:fig1d}. It can be noticed that all the algorithms and even as per the \ac{GAP} bound, we are unable to associate $N_{SC}$ \acp{SC}. This is due to the constraints \eqref{cons1}, \eqref{cons2} and \eqref{cons4}. Thus, the applied backhaul data rate $R = 2.3$ Gbps, bandwidth $B = 0.4$ GHz and number of links $N_l = 10$ limits are less than the system requirements. In Fig. \ref{fig:fig1a}, although the \ac{GAP} bound does not consider the data rate and number of links constraints, 3 \acp{SC} remain un-associated due to the bandwidth limit only. In the \ac{LP}-relaxed solution, if the association matrix entry is greater than zero then the related \ac{NFP}-\ac{SC} pair is shown as associated. However, as \ac{LP}-relaxed solution considers all constraints, 3 \acp{SC} remains un-associated due to limited resources. The B\&B method takes care of all the constraints without any relaxation and thus, 5 \acp{SC} are not associated to satisfy the constraints. Our greedy algorithms also take care of all the constraints and has same number of unassociated SCs as B\&B but works in a different fashion, thus, result in a different association.

To get more insights into the presented algorithms and bounds, we run a few experiments in the following to study the effect of various limitations due to constraints \eqref{cons1}, \eqref{cons2} and \eqref{cons4} on the objective function \eqref{Obj_Fun}, i.e., sum data rate of the overall system. Moreover, we discuss the effect of the constraints on the number of associated \acp{SC}.

\subsection{Experiment 1: Effect of backhaul data rate}\label{Exp1}
Fig. \ref{fig:fig2} plots the sum data rate of the associated \acp{SC} versus $R_r$, which is a ratio of the backhaul data rate limit $R$ to the sum of the demanded data rate by the $N_{SC}$ \acp{SC}. For a single value of $R_r$, we have generated 100 different scenarios and then averaged the associated data rate of \acp{SC}. For various scenarios, the ratio $R_r$ is kept the same by changing the backhaul data rate limit according to the demanded sum data rate of the \acp{SC}. Further, the other limits such as those of the bandwidth and the number of links are relaxed by providing more resources than required. This is done so that the effect of the backhaul data rate can be observed easily. We can notice that for every algorithm except the \ac{GAP} bound, the sum data rate increases with the increase in ratio $R_r$ until the ratio reaches one i.e, $R_r=1$. Then, the sum data rate remains the same even with the increase in $R_r$ because beyond $R_r=1$ algorithm have already associated all the \acp{SC} and thus providing extra resources is unnecessary. Further, we can observe that the performance of our greedy algorithms is very close to the exhaustive B\&B method and \ac{LP} relaxation. The \ac{GAP} bound remains the same as the maximum achievable data rate, as the \ac{GAP} bound does not take care of the constraint \eqref{cons1}.

Fig. \ref{fig:fig3} depicts the effect of the data rate constraint on the number of associated \acp{SC}. It can be noticed that the number of associated \acp{SC} by our greedy algorithms is nearly same as that of \ac{LP}-relaxed solution. The performance of B\&B method is close to the \ac{LP}-relaxed solution, however, it deteriorates when the rate ratio ranges between 0.4 to 0.6\footnote{It must be noted that B\&B is still an upper-bound for our algorithms if for maximizing the sum data rate objective, while in this figure the number of associated \acp{SC} is considered.}. In all conditions the worst performance is obtained by (DM)$^2$S algorithm.

\begin{figure}[t!]\centering
	\setlength\figureheight{6cm}
	\setlength\figurewidth{8.2cm}
	\footnotesize
\definecolor{mycolor1}{rgb}{0.00000,1.00000,1.00000}%
\begin{tikzpicture}
\begin{axis}[%
width=\figurewidth,
height=\figureheight,
xmin=0, xmax=1.5,
xlabel near ticks,
xlabel={$R_r$},
xmajorgrids,
ymin=0, ymax=3,
ylabel near ticks,
ylabel={Sum data rate of SCs (Gbps)},
ymajorgrids,
legend style={at={(0.97,0.03)},anchor=south east,legend cell align=left,align=left,draw=white!15!black}
]
\addplot [color=black,line width=1.5pt,mark size=3.5pt,only marks,mark=square,mark options={solid}]
  table[row sep=crcr]{%
0.1	2.6622\\
0.2	2.6622\\
0.3	2.6622\\
0.4	2.6622\\
0.5	2.6622\\
0.6	2.6622\\
0.7	2.6622\\
0.8	2.6622\\
0.9	2.6622\\
1	2.6622\\
1.2	2.6622\\
1.4	2.6622\\
};
\addlegendentry{GAP Bound};

\addplot [color=red,line width=1.5pt,mark size=4.0pt,only marks,mark=asterisk,mark options={solid}]
  table[row sep=crcr]{%
0.1	0.26622\\
0.2	0.53244\\
0.3	0.79866\\
0.4	1.06488\\
0.5	1.3311\\
0.6	1.59732\\
0.7	1.86354\\
0.8	2.12976\\
0.9	2.39598\\
1	2.6622\\
1.2	2.6622\\
1.4	2.6622\\
};
\addlegendentry{LP};

\addplot [color=blue,line width=1.5pt,mark size=4.0pt,only marks,mark=+,mark options={solid}]
  table[row sep=crcr]{%
0.1	0.2523\\
0.2	0.5196\\
0.3	0.7851\\
0.4	1.0536\\
0.5	1.3233\\
0.6	1.5852\\
0.7	1.8495\\
0.8	2.1192\\
0.9	2.3823\\
1	2.6622\\
1.2	2.6622\\
1.4	2.6622\\
};
\addlegendentry{B\&B};

\addplot [color=green,line width=1.5pt,mark size=4.0pt,only marks,mark=x,mark options={solid}]
  table[row sep=crcr]{%
0.1	0.2124\\
0.2	0.4779\\
0.3	0.7404\\
0.4	1.0161\\
0.5	1.2789\\
0.6	1.5402\\
0.7	1.8105\\
0.8	2.0826\\
0.9	2.3409\\
1	2.6622\\
1.2	2.6622\\
1.4	2.6622\\
};
\addlegendentry{\ref{algo:DMDMSA}};

\addplot [color=red,line width=1.5pt,mark size=4.0pt,only marks,mark=o,mark options={solid}]
  table[row sep=crcr]{%
0.1	0.2523\\
0.2	0.5193\\
0.3	0.7842\\
0.4	1.053\\
0.5	1.3218\\
0.6	1.5825\\
0.7	1.8432\\
0.8	2.1111\\
0.9	2.3691\\
1	2.6622\\
1.2	2.6622\\
1.4	2.6622\\
};
\addlegendentry{\ref{algo:CMDMSA}};

\addplot [color=mycolor1,line width=1.5pt,mark size=5pt,only marks,mark=diamond,mark options={solid}]
  table[row sep=crcr]{%
0.1	0.1737\\
0.2	0.4692\\
0.3	0.747\\
0.4	1.0233\\
0.5	1.2861\\
0.6	1.5444\\
0.7	1.8084\\
0.8	2.0763\\
0.9	2.3442\\
1	2.6622\\
1.2	2.6622\\
1.4	2.6622\\
};
\addlegendentry{(DM)$^2$S};

\end{axis}
\end{tikzpicture}%
	\caption{Sum data rate of \acp{SC} vs. $R_r$ for association algorithms with restrictions $B = 5$ GHz, $N_l = 30$ averaged over 100 different scenarios.}
	\label{fig:fig2}
\end{figure}
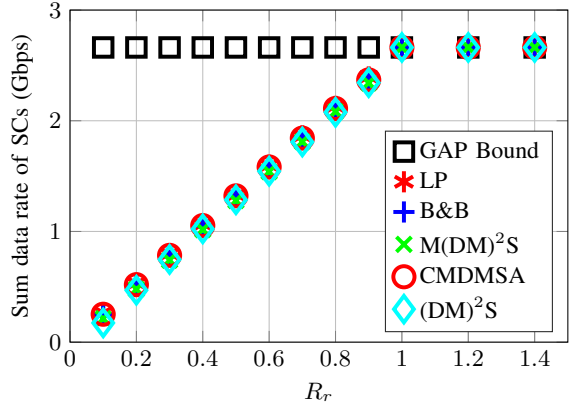

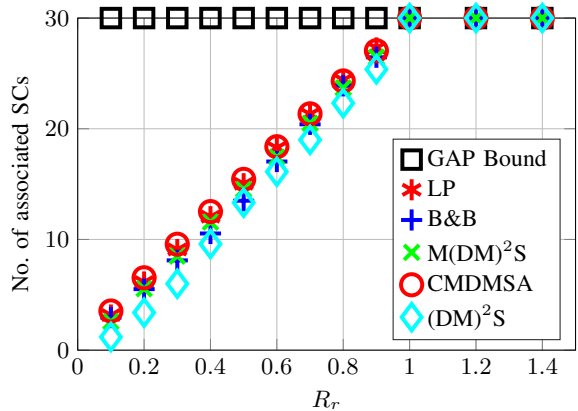
\begin{figure}[t!]\centering
	\setlength\figureheight{6cm}
	\setlength\figurewidth{8.2cm}
	\footnotesize
\definecolor{mycolor1}{rgb}{0.00000,1.00000,1.00000}%
\begin{tikzpicture}
\begin{axis}[%
width=\figurewidth,
height=\figureheight,
xmin=0, xmax=1.5,
xlabel near ticks,
xlabel={$R_r$},
xmajorgrids,
ymin=0, ymax=30,
ylabel near ticks,
ylabel={No. of associated SCs},
ymajorgrids,
legend style={at={(0.97,0.03)},anchor=south east,legend cell align=left,align=left,draw=white!15!black}
]
\addplot [color=black,line width=1.5pt,mark size=3.5pt,only marks,mark=square,mark options={solid}]
  table[row sep=crcr]{%
0.1	30\\
0.2	30\\
0.3	30\\
0.4	30\\
0.5	30\\
0.6	30\\
0.7	30\\
0.8	30\\
0.9	30\\
1	30\\
1.2	30\\
1.4	30\\
};
\addlegendentry{GAP Bound};

\addplot [color=red,line width=1.5pt,mark size=4.0pt,only marks,mark=asterisk,mark options={solid}]
  table[row sep=crcr]{%
0.1	3.172982825821\\
0.2	6.15225153184161\\
0.3	9.09245292827549\\
0.4	12.0585421062843\\
0.5	15.0337367396434\\
0.6	18.2800564547024\\
0.7	21.3222952274393\\
0.8	24.4074773301003\\
0.9	27.2997737115796\\
1	30\\
1.2	30\\
1.4	30\\
};
\addlegendentry{LP};

\addplot [color=blue,line width=1.5pt,mark size=4.0pt,only marks,mark=+,mark options={solid}]
  table[row sep=crcr]{%
0.1	2.94\\
0.2	5.51\\
0.3	8.11\\
0.4	10.54\\
0.5	13.5\\
0.6	17.04\\
0.7	20.4\\
0.8	23.88\\
0.9	26.46\\
1	30\\
1.2	30\\
1.4	30\\
};
\addlegendentry{B\&B};

\addplot [color=green,line width=1.5pt,mark size=4.0pt,only marks,mark=x,mark options={solid}]
  table[row sep=crcr]{%
0.1	2.58\\
0.2	5.55\\
0.3	8.5\\
0.4	11.57\\
0.5	14.52\\
0.6	17.43\\
0.7	20.5\\
0.8	23.7\\
0.9	26.47\\
1	30\\
1.2	30\\
1.4	30\\
};
\addlegendentry{\ref{algo:DMDMSA}};

\addplot [color=red,line width=1.5pt,mark size=4.0pt,only marks,mark=o,mark options={solid}]
  table[row sep=crcr]{%
0.1	3.54\\
0.2	6.54\\
0.3	9.57\\
0.4	12.52\\
0.5	15.43\\
0.6	18.38\\
0.7	21.35\\
0.8	24.32\\
0.9	27.05\\
1	30\\
1.2	30\\
1.4	30\\
};
\addlegendentry{\ref{algo:CMDMSA}};

\addplot [color=mycolor1,line width=1.5pt,mark size=5pt,only marks,mark=diamond,mark options={solid}]
  table[row sep=crcr]{%
0.1	1.19\\
0.2	3.39\\
0.3	6\\
0.4	9.59\\
0.5	13.32\\
0.6	16.12\\
0.7	19\\
0.8	22.31\\
0.9	25.37\\
1	30\\
1.2	30\\
1.4	30\\
};
\addlegendentry{(DM)$^2$S};

\end{axis}
\end{tikzpicture}%
	\caption{Number of associated \acp{SC} vs. $R_r$ for the association algorithms with restrictions $B = 5$ Gbps, $N_l = 30$ averaged over 100 different scenarios.}
	\label{fig:fig3}
\end{figure}

\subsection{Experiment 2: Effect of bandwidth limit}\label{Exp2}
Fig. \ref{fig:fig4} shows the sum data rate of \acp{SC} versus the bandwidth limit applied to \acp{NFP}. We can notice that the sum data rate increases with the increase in bandwidth limit, which is intuitive as \acp{NFP} will be able to associate more \acp{SC} due to more bandwidth. Here, we have supplied more than required backhaul data rate limit $R$ and number of links limit $N_l$, just to observe the effect of bandwidth limit alone. In terms of bounds, we can notice that our derived \ac{GAP} bound provides a more tight bound as compared to \ac{LP}-relaxed bound. In this case, the \ac{GAP} bound and the B\&B method provide the same result due to the fact that both take care of the bandwidth limit and the remaining constraints are already relaxed. This shows that the \ac{GAP} bound whose run-time speed is better than exhaustive B\&B method can be used to study the comparison of the proposed algorithms if we neglect the backhaul data rate and number of links limits. Otherwise, if we include the other two constraints then the \ac{GAP} bound provides an upper limit on the sum data rate of the \acp{SC}. Moreover, we can see the performance of our proposed greedy algorithms \ref{algo:DMDMSA} and \ref{algo:CMDMSA} is very close to the exhaustive B\&B method and better than (DM)$^2$S algorithm.

Fig. \ref{fig:fig5} depicts the effect of the bandwidth constraint on the number of associated \acp{SC}. It can be noticed that the maximum number of \acp{SC} associated by \ac{LP}-relaxed solution matches that achieved by our designed greedy algorithms. In case of less available bandwidth, the \ac{GAP} bound and the B\&B method associated fewer \acp{SC} than the \ac{LP} solution, however the gap in the performance decreases with the increase in bandwidth. The (DM)$^2$S algorithm results in the worst performance.

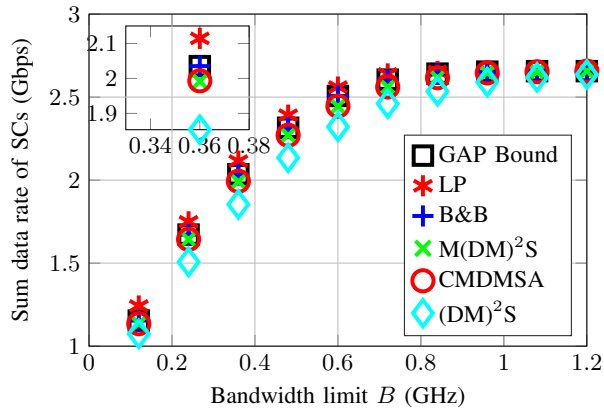
\begin{figure}[t!]\centering
	\setlength\figureheight{6cm}
	\setlength\figurewidth{8.2cm}
	\footnotesize
\definecolor{mycolor1}{rgb}{0.00000,1.00000,1.00000}%
\begin{tikzpicture}
\begin{axis}[%
width=0.2\figurewidth,
height=0.23\figureheight,
at={(0.05965\figurewidth,0.48\figureheight)},
scale only axis,
xmin=0.33,
xmax=0.38,
ymin=1.8534,
ymax=2.15
]
\addplot [color=black,line width=1.5pt,mark size=3.5pt,only marks,mark=square,mark options={solid},forget plot]
  table[row sep=crcr]{%
0.36	2.0352\\
0.48	2.3151\\
};
\addplot [color=red,line width=1.5pt,mark size=4.0pt,only marks,mark=asterisk,mark options={solid},forget plot]
  table[row sep=crcr]{%
0.36	2.11455021017117\\
0.48	2.39117550961658\\
};
\addplot [color=blue,line width=1.5pt,mark size=4.0pt,only marks,mark=+,mark options={solid},forget plot]
  table[row sep=crcr]{%
0.36	2.0352\\
0.48	2.3151\\
};
\addplot [color=green,line width=1.5pt,mark size=4.0pt,only marks,mark=x,mark options={solid},forget plot]
  table[row sep=crcr]{%
0.36	1.9917\\
0.48	2.2701\\
};
\addplot [color=red,line width=1.5pt,mark size=4.0pt,only marks,mark=o,mark options={solid},forget plot]
  table[row sep=crcr]{%
0.36	1.9932\\
0.48	2.2704\\
};
\addplot [color=mycolor1,line width=1.5pt,mark size=5pt,only marks,mark=diamond,mark options={solid},forget plot]
  table[row sep=crcr]{%
0.36	1.8534\\
0.48	2.1345\\
};
\end{axis}

\begin{axis}[%
width=\figurewidth,
height=\figureheight,
xmin=0,
xmax=1.2,
xlabel near ticks,
xlabel={Bandwidth limit $B$ (GHz)},
xmajorgrids,
ymin=1,
ymax=3,
ylabel near ticks,
ymajorgrids,
ylabel={Sum data rate of SCs (Gbps)},
legend style={at={(0.97,0.03)},anchor=south east,legend cell align=left,align=left,draw=white!15!black}
]
\addplot [color=black,line width=1.5pt,mark size=3.5pt,only marks,mark=square,mark options={solid}]
  table[row sep=crcr]{%
0.12	1.1541\\
0.24	1.6719\\
0.36	2.0352\\
0.48	2.3151\\
0.6	2.5068\\
0.72	2.6052\\
0.84	2.6412\\
0.96	2.6532\\
1.08	2.6562\\
1.2	2.6562\\
};
\addlegendentry{GAP Bound};

\addplot [color=red,line width=1.5pt,mark size=4.0pt,only marks,mark=asterisk,mark options={solid}]
  table[row sep=crcr]{%
0.12	1.24203290679145\\
0.24	1.74902206242096\\
0.36	2.11455021017117\\
0.48	2.39117550961658\\
0.6	2.56063282984646\\
0.72	2.63956218776911\\
0.84	2.65449804058225\\
0.96	2.65984766734617\\
1.08	2.66192286907647\\
1.2	2.6622\\
};
\addlegendentry{LP};

\addplot [color=blue,line width=1.5pt,mark size=4.0pt,only marks,mark=+,mark options={solid}]
  table[row sep=crcr]{%
0.12	1.1541\\
0.24	1.6719\\
0.36	2.0352\\
0.48	2.3151\\
0.6	2.5068\\
0.72	2.6052\\
0.84	2.6412\\
0.96	2.6532\\
1.08	2.6562\\
1.2	2.6562\\
};
\addlegendentry{B\&B};

\addplot [color=green,line width=1.5pt,mark size=4.0pt,only marks,mark=x,mark options={solid}]
  table[row sep=crcr]{%
0.12	1.1349\\
0.24	1.6434\\
0.36	1.9917\\
0.48	2.2701\\
0.6	2.4438\\
0.72	2.5608\\
0.84	2.6166\\
0.96	2.6457\\
1.08	2.652\\
1.2	2.6562\\
};
\addlegendentry{\ref{algo:DMDMSA}};

\addplot [color=red,line width=1.5pt,mark size=4.0pt,only marks,mark=o,mark options={solid}]
  table[row sep=crcr]{%
0.12	1.1352\\
0.24	1.6431\\
0.36	1.9932\\
0.48	2.2704\\
0.6	2.4456\\
0.72	2.5608\\
0.84	2.6166\\
0.96	2.6457\\
1.08	2.652\\
1.2	2.6562\\
};
\addlegendentry{\ref{algo:CMDMSA}};

\addplot [color=mycolor1,line width=1.5pt,mark size=5pt,only marks,mark=diamond,mark options={solid}]
  table[row sep=crcr]{%
0.12	1.074\\
0.24	1.5069\\
0.36	1.8534\\
0.48	2.1345\\
0.6	2.3211\\
0.72	2.4603\\
0.84	2.5347\\
0.96	2.5905\\
1.08	2.6184\\
1.2	2.6361\\
};
\addlegendentry{(DM)$^2$S};

\end{axis}
\end{tikzpicture}%
	\caption{Sum data rate of \acp{SC} vs. bandwidth limit $B$ for the association algorithms with restrictions $R = 5$ Gbps, $N_l = 30$ averaged over 100 different scenarios.}
	\label{fig:fig4}
\end{figure}

\begin{figure}[t!]\centering
	\setlength\figureheight{6cm}
	\setlength\figurewidth{8.2cm}
	\footnotesize
\definecolor{mycolor1}{rgb}{0.00000,1.00000,1.00000}%
\begin{tikzpicture}
\begin{axis}[%
width=\figurewidth,
height=\figureheight,
xmin=0, xmax=1.2,
xlabel near ticks,
xlabel={Bandwidth limit $B$ (GHz)},
xmajorgrids,
ymin=10, ymax=30,
ylabel near ticks,
ylabel={No. of associated SCs},
ymajorgrids,
legend style={at={(0.97,0.03)},anchor=south east,legend cell align=left,align=left,draw=white!15!black}
]
\addplot [color=black,line width=1.5pt,mark size=3.5pt,only marks,mark=square,mark options={solid}]
  table[row sep=crcr]{%
0.12	13.24\\
0.24	19.1\\
0.36	22.96\\
0.48	25.97\\
0.6	28.22\\
0.72	29.34\\
0.84	29.81\\
0.96	29.96\\
1.08	30\\
1.2	30\\
};
\addlegendentry{GAP Bound};

\addplot [color=red,line width=1.5pt,mark size=4.0pt,only marks,mark=asterisk,mark options={solid}]
  table[row sep=crcr]{%
0.12	14.2299736443826\\
0.24	19.8767370376576\\
0.36	24.0579367095738\\
0.48	27.214904286794\\
0.6	29.0391353844178\\
0.72	29.7800252382848\\
0.84	29.9191386697171\\
0.96	29.974563085974\\
1.08	29.9976905756372\\
1.2	30\\
};
\addlegendentry{LP};

\addplot [color=blue,line width=1.5pt,mark size=4.0pt,only marks,mark=+,mark options={solid}]
  table[row sep=crcr]{%
0.12	13.24\\
0.24	19.1\\
0.36	22.96\\
0.48	25.97\\
0.6	28.22\\
0.72	29.34\\
0.84	29.81\\
0.96	29.96\\
1.08	30\\
1.2	30\\
};
\addlegendentry{B\&B};

\addplot [color=green,line width=1.5pt,mark size=4.0pt,only marks,mark=x,mark options={solid}]
  table[row sep=crcr]{%
0.12	13.82\\
0.24	19.8\\
0.36	23.54\\
0.48	26.35\\
0.6	28.16\\
0.72	29.18\\
0.84	29.65\\
0.96	29.92\\
1.08	29.97\\
1.2	30\\
};
\addlegendentry{\ref{algo:DMDMSA}};

\addplot [color=red,line width=1.5pt,mark size=4pt,only marks,mark=o,mark options={solid}]
  table[row sep=crcr]{%
0.12	13.82\\
0.24	19.8\\
0.36	23.52\\
0.48	26.35\\
0.6	28.19\\
0.72	29.18\\
0.84	29.65\\
0.96	29.92\\
1.08	29.97\\
1.2	30\\
};
\addlegendentry{\ref{algo:CMDMSA}};

\addplot [color=mycolor1,line width=1.5pt,mark size=5pt,only marks,mark=diamond,mark options={solid}]
  table[row sep=crcr]{%
0.12	11.05\\
0.24	14.82\\
0.36	18.61\\
0.48	22.21\\
0.6	24.68\\
0.72	26.72\\
0.84	27.91\\
0.96	28.79\\
1.08	29.29\\
1.2	29.63\\
};
\addlegendentry{(DM)$^2$S};

\end{axis}
\end{tikzpicture}%
	\caption{Number of associated \acp{SC} vs. bandwidth limit $B$ for the association algorithms with restrictions $R = 5$ Gbps, $N_l = 30$ averaged over 100 different scenarios.}
	\label{fig:fig5}
\end{figure}
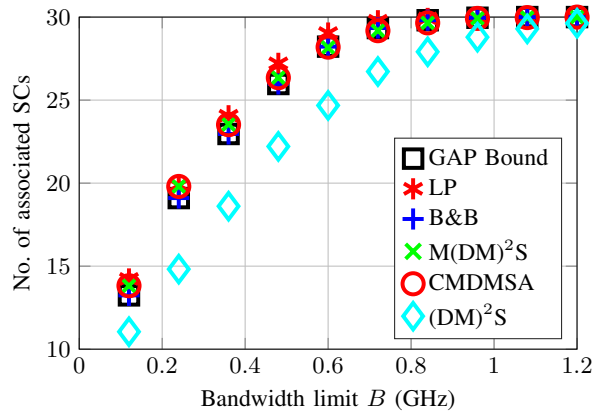

\subsection{Experiment 3: Effect of number of links limit}\label{Exp3}
In Fig. \ref{fig:fig6}, we plot the sum data rate versus the number of links limit $N_l$, where $N_l$ increases from 2 to 30 links. The backhaul data rate $R$ and bandwidth $B$ limits are provided such that they do not effect the association of \acp{SC}, so the effect of the number of links $N_l$ can be observed exclusively. Again, in this case, the \ac{GAP} bound remains unaffected as it does not consider the number of links constraint \eqref{cons4}. There is a performance gap between the exhaustive B\&B method and our presented algorithms in this case. This is due to the fact that our greedy algorithms \ref{algo:DMDMSA} and \ref{algo:CMDMSA} do not take an intelligent decision on the basis of the number of links. For example, there can be a situation, where due to less number of available links, we should give priority to some \acp{SC} to be associated to a certain \ac{NFP}, however our algorithms lack this decision power. Note that, such a situation occurs in a case when very few links $N_l$ are available, as it can be seen from Fig. \ref{fig:fig6} that the performance gap between the B\&B method and our proposed algorithms decreases with the increase in $N_l$. For the case of few available links $N_l$, one can use the (DM)$^2$S algorithm, whose performance is close to the B\&B method.

In Fig. \ref{fig:fig7}, we study the effect of the number of links constraint \eqref{cons4} on the number of associated \acp{SC}. It can be observed that the performance of all the considered algorithms is the same in case the links $N_l$ are either scarce or abundant. The performance of the algorithms varies in a region where the number of links ranges between 5 to 15. \ac{LP}-relaxed solution again leads the algorithms followed by a very close performance between B\&B and our presented greedy algorithms. Meanwhile, the (DM)$^2$S algorithm again results in the minimum number of associated \acp{SC}. The \ac{GAP} bound is not affected by the change in the number of links and provides the upper limit due to the relaxed bandwidth limit.

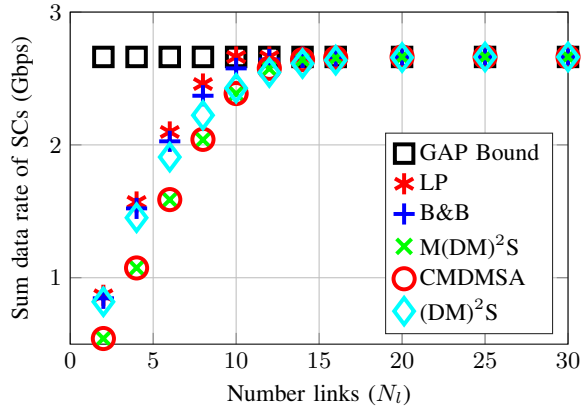
\begin{figure}[t!]\centering
	\setlength\figureheight{6cm}
	\setlength\figurewidth{8.2cm}
	\footnotesize
\definecolor{mycolor1}{rgb}{0.00000,1.00000,1.00000}%
\begin{tikzpicture}
\begin{axis}[%
width=\figurewidth,
height=\figureheight,
xmin=0, xmax=30,
xlabel near ticks,
xlabel={Number links ($N_l$)},
xmajorgrids,
ymin=0.5, ymax=3,
ylabel near ticks,
ylabel={Sum data rate of SCs (Gbps)},
ymajorgrids,
legend style={at={(0.97,0.03)},anchor=south east,legend cell align=left,align=left,draw=white!15!black}
]
\addplot [color=black,line width=1.5pt,mark size=3.5pt,only marks,mark=square,mark options={solid}]
  table[row sep=crcr]{%
2	2.6622\\
4	2.6622\\
6	2.6622\\
8	2.6622\\
10	2.6622\\
12	2.6622\\
14	2.6622\\
16	2.6622\\
20	2.6622\\
25	2.6622\\
30	2.6622\\
};
\addlegendentry{GAP Bound};

\addplot [color=red,line width=1.5pt,mark size=4.0pt,only marks,mark=asterisk,mark options={solid}]
  table[row sep=crcr]{%
2	0.867806039919086\\
4	1.57027317894069\\
6	2.0974500952056\\
8	2.45831563112386\\
10	2.66195328004413\\
12	2.6622\\
14	2.6622\\
16	2.6622\\
20	2.6622\\
25	2.6622\\
30	2.6622\\
};
\addlegendentry{LP};

\addplot [color=blue,line width=1.5pt,mark size=4.0pt,only marks,mark=+,mark options={solid}]
  table[row sep=crcr]{%
2	0.8451\\
4	1.5222\\
6	2.0268\\
8	2.37\\
10	2.5761\\
12	2.6427\\
14	2.6583\\
16	2.6616\\
20	2.6622\\
25	2.6622\\
30	2.6622\\
};
\addlegendentry{B\&B};

\addplot [color=green,line width=1.5pt,mark size=4.0pt,only marks,mark=x,mark options={solid}]
  table[row sep=crcr]{%
2	0.5424\\
4	1.0746\\
6	1.5873\\
8	2.0382\\
10	2.3856\\
12	2.5743\\
14	2.6415\\
16	2.6541\\
20	2.6622\\
25	2.6622\\
30	2.6622\\
};
\addlegendentry{\ref{algo:DMDMSA}};

\addplot [color=red,line width=1.5pt,mark size=4pt,only marks,mark=o,mark options={solid}]
  table[row sep=crcr]{%
2	0.5424\\
4	1.0737\\
6	1.5882\\
8	2.0415\\
10	2.3865\\
12	2.5767\\
14	2.6415\\
16	2.6541\\
20	2.6622\\
25	2.6622\\
30	2.6622\\
};
\addlegendentry{\ref{algo:CMDMSA}};

\addplot [color=mycolor1,line width=1.5pt,mark size=5pt,only marks,mark=diamond,mark options={solid}]
  table[row sep=crcr]{%
2	0.819\\
4	1.4517\\
6	1.9083\\
8	2.223\\
10	2.4267\\
12	2.5449\\
14	2.6127\\
16	2.6382\\
20	2.6577\\
25	2.6619\\
30	2.6622\\
};
\addlegendentry{(DM)$^2$S};

\end{axis}
\end{tikzpicture}%
	\caption{Sum data rate of \acp{SC} vs. number of links limit $N_l$ for the association algorithms with restrictions $R = 5$ Gbps, $B = 5$ GHz averaged over 100 different scenarios.}
	\label{fig:fig6}
\end{figure}

\begin{figure}[t!]\centering
	\setlength\figureheight{6cm}
	\setlength\figurewidth{8.2cm}
	\footnotesize
\definecolor{mycolor1}{rgb}{0.00000,1.00000,1.00000}%
\begin{tikzpicture}
\begin{axis}[%
width=\figurewidth,
height=\figureheight,
xmin=0, xmax=30,
xlabel near ticks,
xlabel={Number links ($N_l$)},
xmajorgrids,
ymin=5, ymax=30,
ylabel near ticks,
ylabel={No. of associated SCs},
ymajorgrids,
legend style={at={(0.97,0.03)},anchor=south east,legend cell align=left,align=left,draw=white!15!black}
]
\addplot [color=black,line width=1.5pt,mark size=3.5pt,only marks,mark=square,mark options={solid}]
  table[row sep=crcr]{%
2	30\\
4	30\\
6	30\\
8	30\\
10	30\\
12	30\\
14	30\\
16	30\\
20	30\\
25	30\\
30	30\\
};
\addlegendentry{GAP Bound};

\addplot [color=red,line width=1.5pt,mark size=4.0pt,only marks,mark=asterisk,mark options={solid}]
  table[row sep=crcr]{%
2	6\\
4	12\\
6	18\\
8	24\\
10	29.991776001471\\
12	30\\
14	30\\
16	30\\
20	30\\
25	30\\
30	30\\
};
\addlegendentry{LP};

\addplot [color=blue,line width=1.5pt,mark size=4.0pt,only marks,mark=+,mark options={solid}]
  table[row sep=crcr]{%
2	5.98\\
4	11.92\\
6	17.6\\
8	22.93\\
10	27.59\\
12	29.38\\
14	29.87\\
16	29.98\\
20	30\\
25	30\\
30	30\\
};
\addlegendentry{B\&B};

\addplot [color=green,line width=1.5pt,mark size=4.0pt,only marks,mark=x,mark options={solid}]
  table[row sep=crcr]{%
2	5.98\\
4	11.92\\
6	17.56\\
8	22.78\\
10	26.83\\
12	28.99\\
14	29.75\\
16	29.91\\
20	30\\
25	30\\
30	30\\
};
\addlegendentry{\ref{algo:DMDMSA}};

\addplot [color=red,line width=1.5pt,mark size=4pt,only marks,mark=o,mark options={solid}]
  table[row sep=crcr]{%
2	5.98\\
4	11.92\\
6	17.56\\
8	22.79\\
10	26.84\\
12	29.02\\
14	29.75\\
16	29.91\\
20	30\\
25	30\\
30	30\\
};
\addlegendentry{\ref{algo:CMDMSA}};

\addplot [color=mycolor1,line width=1.5pt,mark size=5pt,only marks,mark=diamond,mark options={solid}]
  table[row sep=crcr]{%
2	5.94\\
4	11.63\\
6	16.77\\
8	21.19\\
10	24.58\\
12	27.07\\
14	28.67\\
16	29.33\\
20	29.87\\
25	29.99\\
30	30\\
};
\addlegendentry{(DM)$^2$S};

\end{axis}
\end{tikzpicture}%
	\caption{Number of associated \acp{SC} vs. number of links limit $N_l$ for the association algorithms with restrictions $R = 5$ Gbps, $B = 5$ GHz averaged over 100 different scenarios.}
	\label{fig:fig7}
\end{figure}
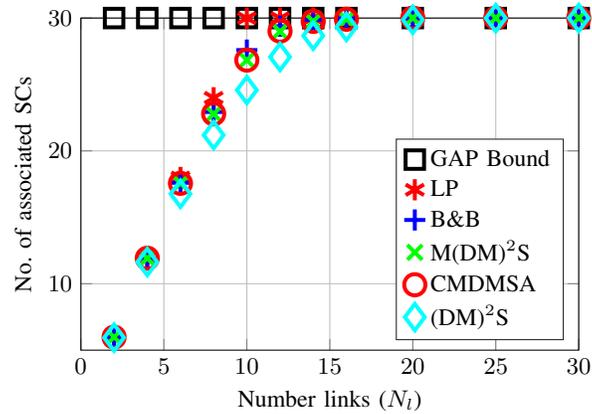

\section{Computational Complexity}\label{sec:Comp_Complexity}
In this section, we provide the worst case complexity and average run time comparison of the algorithms. The B\&B method has the same complexity as that of Brute-force in the worst case \cite{BandB_Algo} and \cite{zhang1996branch}. Table \ref{CompComplex} compares, in terms of the number of flops, the worst case complexity of the presented algorithms with the Brute-force and (DM)$^2$S algorithm. It can be noticed that the presented algorithms are computationally less expensive than the B\&B algorithm in the worst case and slightly more expensive than the (DM)$^2$S algorithm, however, provides the same performance as can be seen from the numerical results.

Table \ref{tab:TimComp} compares the run time speed of the proposed algorithms \ref{algo:DMDMSA} and \ref{algo:CMDMSA} with the \ac{GAP} bound and the exhaustive B\&B method. By observing the Figures \ref{fig:fig2}, \ref{fig:fig4} and \ref{fig:fig6} and as per the constraints \eqref{cons1}, \eqref{cons2} and \eqref{cons4}, all the limits are applied simultaneously. The results are averaged over 100 different scenarios. It can be noticed from Table \ref{tab:TimComp} that the proposed algorithms are faster than the \ac{GAP} bound and B\&B method. Therefore, the proposed algorithms are computationally less expensive and are practically applicable.

\begin{table}[t!]
\renewcommand{\arraystretch}{1.1}
\centering
\caption{Computational complexity of the algorithms.}
\label{CompComplex}
\begin{tabular}{|c|c|}
\hline
\textbf{Algorithm}  &  \textbf{Complexity Order}     \\ \hline
Brute-force         &  $\mathcal{O}\left( N_D^3 N_{SC}^{N_D + 2} \right)$      \\ \hline
\ref{algo:DMDMSA}   &  $\mathcal{O}\left( N_D^2 N_{SC}^2 \right)$               \\ \hline
\ref{algo:CMDMSA}   &  $\mathcal{O}\left( N_D^2 N_{SC}^2 \right)$               \\ \hline
(DM)$^2$S           &  $\mathcal{O}\left( N_D N_{SC} \right)$                   \\ \hline
\end{tabular}
\end{table}

\begin{table}[t!]
\renewcommand{\arraystretch}{1.1}
\centering
\caption{Run Time comparison of the algorithms for the system with constraints $R = 2.3$ Gbps, $B = 0.6$ GHz and $N_l = 8$ averaged over 100 different scenarios.}
\label{tab:TimComp}
\begin{tabular}{|c|c|}
\hline
\textbf{Method}          & \textbf{Elapsed Time (seconds)}   \\ \hline
\ac{GAP} bound           & 0.036                             \\ \hline
B\&B                     & 45.145                            \\ \hline
\ref{algo:DMDMSA}        & 0.000445                          \\ \hline
\ref{algo:CMDMSA}        & 0.000523                          \\ \hline
\end{tabular}
\end{table}

\section{Conclusions}\label{sec:Conc}
This work considered the use of \acp{NFP} as fronthaul hubs to provide backhaul connectivity to an ultra dense network of \acp{SC}. An association problem of \acp{SC} and \acp{NFP} is formulated in order to maximize the sum data rate of the overall \ac{SC} network along with the consideration of backhaul and \ac{NFP}-related limitations such as backhaul data rate, bandwidth and number of links limits of \ac{NFP}. In the literature, the association problem of \acp{SC} and \acp{NFP} was claimed to be NP-hard, however, in this work, we have shown it to be NP-hard by relating it with the \ac{GAP}. Then, using this relevance, a performance bound is derived and verified by a numerical comparison with LP relaxed solution as well as B\&B exhaustive search. Further, two efficient (less complex) greedy algorithms are designed to solve the problem. The performance of the presented algorithms is same as of exhaustive B\&B search as well as the presented bounds, both in terms of sum data rate and number of associated \acp{SC}. However, the presented algorithms has a lower complexity than all the other methods, thus they can be practically implemented. In future work, we will consider another aspect of the problem where the objective is to serve maximum number of users, i.e., network centric approach. Similar techniques can be used to derive related bounds and also similar greedy algorithms can be presented for such a case.

\balance
\bibliographystyle{IEEEtran}
\bibliography{JSAC_Airbone_SC_NFP}

\begin{thebibliography}{10}
\providecommand{\url}[1]{#1}
\csname url@samestyle\endcsname
\providecommand{\newblock}{\relax}
\providecommand{\bibinfo}[2]{#2}
\providecommand{\BIBentrySTDinterwordspacing}{\spaceskip=0pt\relax}
\providecommand{\BIBentryALTinterwordstretchfactor}{4}
\providecommand{\BIBentryALTinterwordspacing}{\spaceskip=\fontdimen2\font plus
\BIBentryALTinterwordstretchfactor\fontdimen3\font minus
  \fontdimen4\font\relax}
\providecommand{\BIBforeignlanguage}[2]{{%
\expandafter\ifx\csname l@#1\endcsname\relax
\typeout{** WARNING: IEEEtran.bst: No hyphenation pattern has been}%
\typeout{** loaded for the language `#1'. Using the pattern for}%
\typeout{** the default language instead.}%
\else
\language=\csname l@#1\endcsname
\fi
#2}}
\providecommand{\BIBdecl}{\relax}
\BIBdecl

\bibitem{pizzinat2015}
A.~Pizzinat, P.~Chanclou, F.~Saliou, and T.~Diallo, ``Things you should know
  about fronthaul,'' \emph{IEEE J. Lightw. Tech.}, vol.~33, no.~5, pp.
  1077--1083, Mar. 2015.

\bibitem{chih2014toward}
I.~Chih-Lin, C.~Rowell, S.~Han, Z.~Xu, G.~Li, and Z.~Pan, ``Toward green and
  soft: a {5G} perspective,'' \emph{IEEE Commun. Mag.}, vol.~52, no.~2, pp.
  66--73, 2014.

\bibitem{Ericsson_m2020}
\BIBentryALTinterwordspacing
``Microwave towards 2020,'' Ericsson, Stockholm, Sweden, White Paper, \notype,
  Sep. 2015, accessed on May 18, 2017. [Online]. Available:
  \url{https://www.ericsson.com/res/docs/2015/microwave-2020-report.pdf}
\BIBentrySTDinterwordspacing

\bibitem{SC_virt15}
\BIBentryALTinterwordspacing
``Buisness case elements for small cell virtualization,'' Real Wireless Ltd.,
  Pulborough, U.K., White Paper SCF 158, \notype, Jun. 2015, accessed on May
  18, 2017. [Online]. Available: \url{http://scf.io/en/download.php?doc=158}
\BIBentrySTDinterwordspacing

\bibitem{ShakirFSOMAG}
M.~Alzenad, M.~Z. Shakir, H.~Yanikomeroglu, and M.-S. Alouini, ``{FSO}-based
  vertical backhaul/fronthaul framework for {5G+} wireless networks,''
  \emph{arXiv preprint arXiv:1607.01472}, 2017.

\bibitem{ahmadi2017novel}
H.~Ahmadi, K.~Katzis, and M.~Z. Shakir, ``A novel airborne self-organising
  architecture for {5G+} networks,'' in \emph{IEEE VTC-Fall 2017}, Sep. 2017.

\bibitem{ATGmodel}
A.~Al-Hourani, S.~Kandeepan, and A.~Jamalipour, ``Modeling air-to-ground path
  loss for low altitude platforms in urban environments,'' in \emph{IEEE
  GLOBECOM}, pp. 2898--2904, Dec. 2014.

\bibitem{ATG_optDrone1}
A.~Al-Hourani, S.~Kandeepan, and S.~Lardner, ``Optimal {LAP} altitude for
  maximum coverage,'' \emph{IEEE Wireless Commun. Lett.}, vol.~3, no.~6, pp.
  569--572, Dec. 2014.

\bibitem{TwoDrones}
M.~Mozaffari, W.~Saad, M.~Bennis, and M.~Debbah, ``Drone small cells in the
  clouds: Design, deployment and performance analysis,'' in \emph{IEEE
  GLOBECOM}, pp. 1--6, Dec. 2015.

\bibitem{IremOneDrone}
R.~I. Bor-Yaliniz, A.~El-Keyi, and H.~Yanikomeroglu, ``Efficient {3-D}
  placement of an aerial base station in next generation cellular networks,''
  in \emph{IEEE ICC}, pp. 1--5, May. 2016.

\bibitem{ElhamBackhaul}
E.~Kalantari, M.~Z. Shakir, H.~Yanikomeroglu, and A.~Yonga{\c{c}}oglu,
  ``Backhaul-aware robust {3D} drone placement in {5G+} wireless networks,'' in
  \emph{IEEE ICC}, 2017.

\bibitem{alzenad20173d}
M.~Alzenad, A.~El-Keyi, F.~Lagum, and H.~Yanikomeroglu, ``{3D} placement of an
  unmanned aerial vehicle base station {(UAV-BS)} for energy-efficient maximal
  coverage,'' \emph{IEEE Wireless Communications Letters}, vol.~6, no.~4, pp.
  434--437, Aug. 2017.

\bibitem{ElhamMultiPSO}
E.~Kalantari, H.~Yanikomeroglu, and A.~Yongacoglu, ``On the number and {3D}
  placement of drone base stations in wireless cellular networks,'' in
  \emph{IEEE VTC}, pp. 1--6, Sep. 2016.

\bibitem{Mozaffari2016}
M.~Mozaffari, W.~Saad, M.~Bennis, and M.~Debbah, ``Efficient deployment of
  multiple unmanned aerial vehicles for optimal wireless coverage,'' \emph{IEEE
  Commun. Lett.}, vol.~20, no.~8, pp. 1647--1650, Aug. 2016.

\bibitem{Sharma2016}
V.~Sharma, K.~Srinivasan, H.-C. Chao, K.-L. Hua, and W.-H. Cheng, ``Intelligent
  deployment of {UAVs} in {5G} heterogeneous communication environment for
  improved coverage,'' \emph{J. of Network and Computer Applications}, pp.~--,
  2016.

\bibitem{mozaffari2017OTT}
M.~Mozaffari, W.~Saad, M.~Bennis, and M.~Debbah, ``Optimal transport theory for
  cell association in {UAV}-enabled cellular networks,'' \emph{IEEE
  Communications Letters}, vol.~21, no.~9, pp. 2053--2056, Sep. 2017.

\bibitem{Mozaffari2017OTTTrans}
------, ``Wireless communication using unmanned aerial vehicles {(UAVs)}:
  Optimal transport theory for hover time optimization,'' \emph{IEEE Trans. on
  Wireless Commun.}, vol.~PP, no.~99, pp. 1--1, Sep. 2017.

\bibitem{kalantari2017user}
E.~Kalantari, I.~Bor-Yaliniz, A.~Yongacoglu, and H.~Yanikomeroglu, ``User
  association and bandwidth allocation for terrestrial and aerial base stations
  with backhaul considerations,'' in \emph{IEEE PIMRC}, Oct. 2017.

\bibitem{AwaisGlobeCom}
S.~A.~W. Shah, T.~Khattab, M.~Z. Shakir, and M.~O. Hasna, ``A distributed
  approach for networked flying platform association with small cells in {5G+}
  networks,'' in \emph{IEEE GLOBECOM}, Dec. 2017.

\bibitem{AwaisPIMRC}
------, ``Association of networked flying platforms with small cells for
  network centric {5G+ C-RAN},'' in \emph{IEEE PIMRC}, Oct. 2017.

\bibitem{rossBound}
G.~T. Ross and R.~M. Soland, ``A branch and bound algorithm for the generalized
  assignment problem,'' \emph{Mathematical programming}, vol.~8, no.~1, pp.
  91--103, Dec. 1975.

\bibitem{matern2013}
B.~Mat{\'e}rn, \emph{Spatial variation}, ser. Springer Lecture Notes in
  Statistics.\hskip 1em plus 0.5em minus 0.4em\relax Springer, 1986, vol.~36.

\bibitem{martello1990knapsack}
S.~Martello, ``Knapsack problems,'' \emph{Algorithms and Computer
  Implementation}, 1990.

\bibitem{fisher1986}
M.~L. Fisher, R.~Jaikumar, and L.~N. Van~Wassenhove, ``A multiplier adjustment
  method for the generalized assignment problem,'' \emph{Management Science},
  vol.~32, no.~9, pp. 1095--1103, 1986.

\bibitem{NPhard1}
E.~Karamad, R.~S. Adve, Y.~Lostanlen, F.~Letourneux, and S.~Guivarch,
  ``Optimizing placements of backhaul hubs and orientations of antennas in
  small cell networks,'' in \emph{IEEE ICCW}, pp. 68--73, Jun. 2015.

\bibitem{NPhard2}
Z.~Mlika, E.~Driouch, W.~Ajib, and H.~Elbiaze, ``A completely distributed
  algorithm for user association in {HetSNets},'' in \emph{IEEE ICC}, pp.
  2172--2177, Jun. 2015.

\bibitem{BandB_Algo}
A.~Schrijver, \emph{Theory of linear and integer programming}.\hskip 1em plus
  0.5em minus 0.4em\relax New York, NY, USA: John Wiley \& Sons, 1986.

\bibitem{zhang1996branch}
W.~Zhang, ``Branch and {Bound} search algorithms and their computational
  complexity.'' DTIC Document, Tech. Rep., 1996.

\end{thebibliography}
\end{document}